\documentclass[english, 12pt, letter]{article}
\usepackage{amsmath,amssymb,amsthm}
\usepackage{array}
\usepackage{mathtools}
\usepackage{natbib}
\usepackage[nodisplayskipstretch]{setspace}
\usepackage{bm}
\usepackage{threeparttable} % FOR NOtes
\usepackage{longtable}
\usepackage{graphicx}
\usepackage{enumerate}
\usepackage{commath}			% for \norm{...}
\usepackage{bibentry}
\usepackage{booktabs}
\usepackage[titletoc,title]{appendix}
\usepackage{multirow}
\usepackage{arydshln}
\usepackage{tabularx}
\usepackage{dsfont}
\usepackage{rotating} % Rotating table
\usepackage[font={small}]{subcaption}
\usepackage[left=1in, right=1in, top=1in, bottom=1in]{geometry}

\usepackage[colorinlistoftodos,textsize=small]{todonotes}
\allowdisplaybreaks

\usepackage{url}
\usepackage{xr-hyper}  % for cross-referencing
\usepackage[
	breaklinks,
	colorlinks=true,
	pagebackref=false,
	pdfauthor={},%
	pdftitle={},% 
	citecolor=blue,
	linkcolor=blue,
	urlcolor=blue
	]{hyperref}

\newtheorem{thm}{Theorem}

\newtheoremstyle{remark2}{1ex}{1ex}%
      {}% Body font
      {}% Indent amount (empty = no indent, \parindent = para indent)
      {\bf}% Thm head font. notice \YYshape command works unlike \textyy
      {.}% Punctuation after thm head
      {5pt}% Space after thm head (\newline = linebreak)
      {\thmname{#1}\thmnumber{ #2}\thmnote{ \slshape{(#3)}}} % Thm head spec, so numbering first
\theoremstyle{remark2}
\newtheorem{rem}{Remark}

\newtheoremstyle{remark3}{1ex}{1ex}%
      {\it}% Body font
      {}% Indent amount (empty = no indent, \parindent = para indent)
      {\bf}% Thm head font. notice \YYshape command works unlike \textyy
      {.}% Punctuation after thm head
      {5pt}% Space after thm head (\newline = linebreak)
      {\thmname{#1}\thmnumber{.#2}\thmnote{ \slshape{(#3)}}} % Thm head spec, so numbering first
\theoremstyle{remark3}

\newcommand{\1}{\mathds{1}}
\DeclareMathOperator{\essinf}{ess\,inf}

\usepackage{babel}
\input{ee.sty}

\@ifundefined{bibfont}{  } %%%% if then else
    {  }

\makeatletter
\renewenvironment{proof}[1][\bfseries\proofname]{\par
   \pushQED{\qed}%
   \normalfont \topsep6\p@\@plus6\p@\relax
   \trivlist
   \item[\hskip\labelsep
     #1\@addpunct{:}]\ignorespaces
}{%
   \popQED\endtrivlist\@endpefalse
}
\makeatother

\newcommand{\Comments}{1}
\newcommand{\mynote}[2]{\ifnum\Comments=1\textcolor{#1}{#2}\fi}
\newcommand{\mytodo}[2]{\ifnum\Comments=1%
  \todo[linecolor=#1!80!black,backgroundcolor=#1,bordercolor=#1!80!black]{#2}\fi}

\ifnum\Comments=1               % fix margins for todonotes
\fi

\ifnum\Comments=1               % fix margins for todonotes
\fi

\newcommand{\CRPS}{\operatorname{CRPS}}
\newcommand{\QSR}{\operatorname{QSR}}

\newcommand{\D}{\,\mathrm{d}}

\renewcommand{\E}{\mathbb{E}}
\renewcommand{\P}{\mathbb{P}}

\begin{document}

\baselineskip18pt
\renewcommand\floatpagefraction{.9}
\renewcommand\topfraction{.9}
\renewcommand\bottomfraction{.9}
\renewcommand\textfraction{.1}
\setcounter{totalnumber}{50}
\setcounter{topnumber}{50}
\setcounter{bottomnumber}{50}
\abovedisplayskip1.5ex plus1ex minus1ex
\belowdisplayskip1.5ex plus1ex minus1ex
\abovedisplayshortskip1.5ex plus1ex minus1ex
\belowdisplayshortskip1.5ex plus1ex minus1ex

\title{
Time preference effects in forecasting\thanks{YH gratefully acknowledges support of the Deutsche Forschungsgemeinschaft (DFG, German Research Foundation) through grants 460479886 and 531866675. We are grateful to Metaculus for granting access to their data, and we especially thank Nikos Bosse for his overall support and assistance with data transformation. We are also indebted to seminar participants at Goethe University Frankfurt, TU Chemnitz, and the Forecasting Research Institute.}
}

\author{
	Yannick Hoga\thanks{Faculty of Economics and Business Administration, University of Duisburg-Essen, Universit\"atsstra\ss e 12, D--45117 Essen, Germany, \href{mailto:yannick.hoga@vwl.uni-due.de}{yannick.hoga@vwl.uni-due.de}.}
	\and 
		Niklas V.~Lehmann\thanks{Faculty of Business Administration, TU Bergakademie Freiberg, Schlossplatz 1, D--09599 Freiberg, Germany, \href{mailto:Niklas-Valentin.Lehmann@vwl.tu-freiberg.de}{Niklas-Valentin.Lehmann@vwl.tu-freiberg.de}.}
}

\date{\today}
\maketitle

\begin{abstract}
	\noindent 
	We study the evaluation of forecasts regarding the timing and occurrence of uncertain future events, such as volcanic eruptions, the start of a war or the beginning of a recession. We show theoretically that a typical approach---evaluating the forecasts after the event occurred---incentivizes dishonest predictions if forecasters discount future rewards in favor of more immediate benefits. An empirical application to forecasting tournament data finds strong empirical evidence that forecasters adjust predictions in response to these incentives, implying that existing forecasts of such events are likely systematically overstating the probability of early occurrence. We conclude that rewarding such forecasts in an incentive-compatible way is inherently challenging.\\
	
\end{abstract}

\doublespacing

\section{Introduction}

When collecting expectations and forecasts regarding disruptive future events, such as wars, technological breakthroughs, and volcanic eruptions, we would like to reward forecasters for their accurate assessment of the matter.\footnote{%For differentiation
The idea that we can create incentives for careful and honest assessments of probability judgments dates back at least to \citet{brier_verification_1950}. 
He found that there exist incentive-compatible payment schemes that provide the most money to forecasters who honestly share their expectation regarding future outcomes.}
However, it is difficult to incentivize forecasters for accurately sharing their expectation when we do not know when, or if, the event will occur. 
Specifically, it is common practice to evaluate forecasts after outcomes are known, which may distort incentives to report honestly  
if human forecasters discount potential rewards in the distant future relative to more immediate ones.

To illustrate, let us consider the question: When will there be a volcanic eruption, which ejects over 100 cubic kilometers of erupted material? 
An expert volcanic forecaster will have an \textit{expectation} regarding this event, assigning each future point in time a probability of such an eruption occurring. Assume we wanted to elicit this expectation, e.g., using a survey or an interview. We would want to make sure that the forecaster reports their honest expectation. Ideally, the forecaster would even refine their expectation by engaging in research relevant to the prediction task. How should such a forecast be evaluated and rewarded? Evaluation refers to determining whether the forecast has been ``good'' in the most general sense. This requires \textit{either} that the eruption has happened, upon which we can evaluate whether the forecast was indeed indicative of the eruption, \textit{or} that a certain time has passed, say 20 years without any such eruption. The problem is that most forecasters will likely care less about evaluations in 20 years time, and more about those in the short-term future. Since any evaluation of their forecast in the short-term future can \textit{only} be extraordinarily positive in the case of such an eruption (where the prediction will get attention and evaluation), the forecaster has an \textit{incentive} to report a higher probability of an early eruption than they truly expect there to be.

Against this backdrop, we formally investigate the incentive to misreport one's expectation when forecasters discount the future \textit{and} the predicted event can occur at arbitrary points in time, such that the timing of the reward is uncertain.
Our first main contribution is to show that \textit{every} incentive-compatible payment scheme that rewards accurate forecasts becomes incentive-\textit{in}compatible with honest reporting when forecasters discount the future and---as is common and often unavoidable---are evaluated upon occurrence.
Our second main contribution is to empirically investigate whether real human forecasters misreport their own expectations using data from a real-world forecasting tournament.
As predicted by our theory, we find strong evidence that forecasters respond to incentives to misreport their expectation. In discussing four potential solutions to this problem, we find that it is inherently challenging to provide incentives for honest reports when the predicted event can occur at any future time.

The hypothesis that we investigate in this article may strike readers as an unintuitive one. Why would forecasters be misreporting their own expectation?
There is a literature on probabilistic judgment in experimental economics, statistical decision theory and related fields that provides ample evidence of forecasters misreporting their expectation when given incentives to do so.
For example, \citet{armantier_eliciting_2013} find in a series of experiments that participants are susceptible to opportunities to hedge or avert risk when making probabilistic judgments.  Another reason to actively change the information reported via forecasts arises in competition with other forecasters \citep{ottaviani_strategy_2006}. Forecasters may try to actively stand out from the crowd \citep{witkowski_incentive-compatible_2023} or not be the one with the worst prediction, i.e., to make predictions more similar to the crowd in spite of contradictory private information.

Despite this large amount of literature \citep[reviewed by][]{charness_experimental_2021}, we are not aware of a single work that studies the issue of time preferences in forecasts. The only explicit mention of time preferences related to forecasts that we know of appears in \citet{chambers_dynamic_2021}, who state that ``time preferences [...] complicate the task of elicitation''.
The main reason for this seems to be that the subtle---yet important---distinction between scheduled and unscheduled future events has not been made in the past.
Historically speaking, the literature on forecasting and belief elicitation has mostly focused on events which are certain, or almost certain, to occur \textit{at one point in time} or form a time series \citep{zellner_survey_2021}. 
Henceforth, we refer to events that are unscheduled and can happen at multiple points in time as \textit{time-varying events}, and others as \textit{time-fixed events}, where resolution occurs at some scheduled (fixed) time in the future.

Time-fixed events are prevalent in macroeconomics (GDP growth next quarter, inflation one year ahead), meteorology (rainfall tomorrow, temperature next week) or finance (Value-at-Risk 10-days-ahead, interest rates in one month).
All of these events have in common that they are scheduled, that is, their verification is fixed in time.
However, a great number of events for which we would like to gather forecasts are not scheduled and can happen at arbitrary points in time. Examples include floods, volcanic eruptions, deaths, insolvency, technological breakthroughs, wars and whether a team will be eliminated in the playoffs. 
As \citet{vere-jones_forecasting_1995} puts it: ``[...] forecasting earthquakes differs from most routine forecasting problems in that it deals not with a discrete or continuous time series, but with sudden random events. But, although problems of this kind may be relatively uncommon in forecasting practice, they are certainly not unknown. In essence, they fall into the same category as forecasting lifetimes, for which there exist substantial literatures in the engineering (reliability), actuarial, and medical contexts.''

Indeed, \textit{Survival Analysis}---also known as reliability analysis, duration analysis or event history analysis---is an entire subfield which is predominantly concerned with the timing and occurrence probability of one-time events \citep{allison_event_2014}. 
Whilst there is overlap with the contents of this article, survival analysis deals less with forecasting and more with inference.
Usually, survival analysis is concerned with estimating the survival of a group of subjects (hence the name), e.g., estimating the life expectancy of a population. Furthermore, these estimates are often conditional on observed covariates, such as age and gender. Survival analysis may also include causal inference, such as inferring the effect of medication on life expectancy. 

Although there is some relation of our work to survival analysis, the most closely related paper to ours in spirit is \citet{Lea17}.
Similar to us, they also take a decision-theoretic perspective on the task of forecasting---understanding the forecaster as a self-interested and rational reporter of private information.
Specifically, \citet{Lea17} describe the Forecaster's Dilemma, where forecasts get outsized attention if, and only if, extreme events happen, such as after the 2008 financial crisis and the COVID-19 pandemic. This is equivalent to weighting the evaluation of forecasts conditional on specific outcomes, such as the occurrence of a crisis. \citet{GR11} show that any such weighting will inevitably distort optimal forecasts. In the context of \citet{Lea17}, forecasters that always predict calamity will be perceived as more accurate than skilled forecasters.
At their core, the Forecaster's Dilemma and the problem we investigate in this paper share the same underlying structure: in both cases, the evaluation of a forecast is not weighted equally across all possible real-world outcomes. Instead, certain outcomes cause the evaluation to matter more than others.
In the Forecaster's Dilemma, this imbalance arises externally---extreme outcomes attract disproportionate public attention to the forecast. In the setting we describe, the imbalance arises internally---forecasters who discount the future place greater weight on more immediate evaluations, and thus, more immediate outcomes.

The remainder of this article proceeds as follows. Section~\ref{sec:Binary_forecasts} investigates the incentive-compatibility of strictly proper scoring rules for binary probabilistic forecasts. 
Section~\ref{sec:Continuous-time forecasts} is concerned with the incentive-compatibility of the Continuous Ranked Probability Score, an error measure for distributional forecasts. Section~\ref{sec:Empirical analysis} features the empirical analysis of forecasters in a real-world forecasting tournament. Section~\ref{sec:Discussion} discusses the broader relevance of the results and outlines future research directions. Finally, Section~\ref{sec:Conclusion} concludes.

\section{Predicting the Occurrence of Future Events}\label{sec:Binary_forecasts}

We investigate the incentive to misreport by employing an illustrative example and leave a more formal investigation to Section~\ref{sec:Continuous-time forecasts}.
Let two players, Alice and Bob, compete against each other in a best-of-three series. The first player to achieve two victories wins the series. Alice has won the first matchup, and thus the standing is 1-0. Potential outcomes are displayed in Figure~\ref{fig:Game}.
We are interested in the probability of Alice winning the entire series. We therefore turn to a forecaster and ask: 

\begin{figure}
	\centering
		\includegraphics[width=0.5\textwidth]{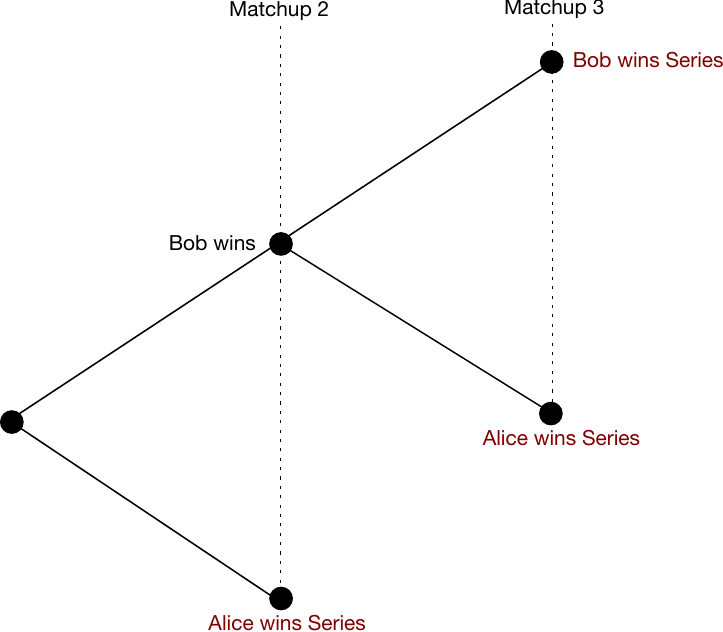}
    	\caption{Potential outcomes in the best-of-three series}\label{fig:Game}
\end{figure}

\begin{quote}
    Will Alice win the series?
\end{quote}

We denote the event of interest (that Alice wins the series) by $X \in \{0,1\}$.
Here, $X=1$ corresponds to the case where Alice has won the series. Let $T \in \{2,3\}$ indicate when Alice has won the series; $T=2$ is the case where Alice has won in matchup 2, and $T=3$ is the case where Alice has won in matchup 3. The event of Alice winning the series is binary and the forecast is a probability $G\in[0,1]$.
Throughout the paper, we use the terms report, forecast and prediction interchangeably.

In order to incentivize the forecaster to report honestly and accurately, we evaluate the forecast with a \textit{strictly proper scoring rule}\footnote{Strictly proper scoring rules are a widely employed tool to evaluate forecasts. Such scoring rules are functions with the property that the expected score is minimized \textit{only} by setting the reported variable $G$ equal to the expected outcome $\E[X]$ \citep{GR07}. Therefore, if a forecaster issues a forecast $G$ and gets a fixed reward from which we subtract the score $S(G,X)$ when $X$ materializes, then---on average---she can do no better than to issue the true forecast $G=\E[X]$ to maximize her expected reward (or minimize the expected score).
In this sense, the score incentivizes honest forecasts or, more lyrically, serves as a ``truth serum''.} $S(G,X)$ and reward the forecaster proportionally to the negated score as soon as the outcome $X$ is known.
The forecaster is a risk-neutral score minimizer and discounts the future.\footnote{That is, we ignore risk aversion throughout this analysis. We conjecture that the problem remains even without risk-neutrality. However, risk aversion is then an additional reason why the forecaster may report a dishonest forecast \citep{hossain_binarized_2013}.} Without loss of generality, we assume that the forecaster discounts between matchups 2 and 3 at rate $r>0$; that is, a payoff realized after matchup 2 is valued $(1+r)$-times as much as the same payoff realized after matchup 3.

The forecaster now has an incentive to misreport their own true expectation $\E[X]$. 
This is because the forecaster gives extra weight ($1+r$) to the case where Alice wins the series in matchup 2, inflating the reported probability that Alice will win the series. Next, we show that the error-minimizing forecast $G$ is strictly greater than the forecasters true expectation $\E[X]$.
We can integrate the discount rate $r$ directly into the forecaster's valuation of her expected score and call this a \textit{discounted scoring rule} $S_d$:%\footnote{For our purposes, it does not matter whether we model Chloe's score as a net gain or loss to her, as we can add\yannick{but this is about multiplication by $-1$!} any arbitrary constant reward to Chloe's payout without affecting her incentives \citep{GR07}.}
\begin{align*}
	S_d(G, X) =\begin{cases}
	S(G,X) \times (1+r), & \text{if}\ T=2,\\
	S(G,X), & \text{else.}
	\end{cases}
\end{align*}

We denote the true expectation of Alice winning the series as $\E[X] = \P\{T=2\} + \P\{T=3\}$, the sum of probabilities of Alice winning either in match 2 or 3.
If the forecaster were to minimize the expected undiscounted score, then they could do no better than to issue $G=\E[X]$ (because this maximizes the expected payoff, $-\E[S(G,X)]$).
However, because of the forecasters present bias, they minimize $\E[S_d(G,X)]$, issuing a different forecast.\footnote{Throughout the paper, we simply take ``present bias'' to mean that earlier payoffs are preferred to later payoffs of the same amount. However, there is a large literature in economics where the term is understood more narrowly \citep{OR15,Cha21}. Specifically, in behavioral economics, present bias (or also: the immediacy effect) refers to the tendency to favor a smaller reward available immediately over a larger reward received later, while reversing this preference when both rewards are shifted equally into the future.}
Specifically, by straightforward computations, the forecaster's expected score is 
\begin{align}
    \E[S_d(G,X)\mid X=0] & = S(G,0), \notag\\
    \E[S_d(G,X) \mid X=1] & = S(G,1) \times  (1+\P\{T=2 \mid X=1\}r).\label{eq:DS}
\end{align}
The added factor in \eqref{eq:DS} is essentially an additional constant weight.
This suffices to distort incentives for honest reporting as multiplying any strictly proper scoring rule with a constant factor conditional on outcomes, i.e., putting more weight on one outcome, leads to a predictably improper scoring rule. For the proof, we refer to Lemma 4 in \citet{lindley_scoring_1982}, which---in the context of forecasting---is accessibly explained by \citet[p.~197]{parmigiani_decision_2009}. The optimal forecast is no longer the true expectation.

We illustrate this result by demonstrating the improperness of the discounted quadratic scoring rule ($\QSR_d$), where the quadratic scoring rule (or also: Brier score)
\begin{equation}\label{eq:(QSR)}
	\QSR(G,X) = (X-G)^2
\end{equation}
is one specific strictly proper scoring rule. 
The expected score then is
\begin{align}
    \E[\QSR_d(G,X) \mid X=0] & = G^2  \label{eq:dS^2_1},  \\ 
    \E[\QSR_d(G,X) \mid X=1]    &= (1-G)^2 \times (1+\P\{T=2 \mid X=1\}r).\label{eq:dS^2_2}
\end{align}

From Lemma 4 of \citet{lindley_scoring_1982} we know that the relationship between the expectation $\E[X]$ and the reported forecast $G$ (that minimizes $\E[S_d(G,X)]$) is given by
\begin{equation*}
    \mathbb{E}[X] = \frac{\frac{\partial}{\partial G} (G^2)}{\frac{\partial}{\partial G} (G^2) - \frac{\partial}{\partial G} \big[ (1-G)^2 \times (1+\P\{T=2 \mid X=1\}r)\big]}.
\end{equation*}
Solving for $G$ and re-arranging gives that
\begin{equation}
    G = \frac{ \E[X]+\P\{T=2\}r}{1+\P\{T=2\} r}.\label{eq:Chloe_pred}
\end{equation}

It is obvious that the error-minimizing forecast $G$ is strictly greater than $\E[X]$ for $r>0$. If either time preferences $r$ or the possibility for the event to occur early $\P\{T=2\}$ is set to zero, then \eqref{eq:Chloe_pred} yields that the error-minimizing report is $G=\E[X]$.

Figure~\ref{fig:Chloe_calibration} plots the optimal forecast on the horizontal axis and the true expectation on the vertical axis, setting $\P\{T=2\}= 0.2$. The blue lines represent the score-minimizing forecast for different $r>0$. The red line corresponds to the 45 degree diagonal, where $G=\E[X]$. We see that the reported probability increases with a rising discount factor.

\begin{figure}[t!]
	\centering
		\includegraphics[width=0.6\textwidth]{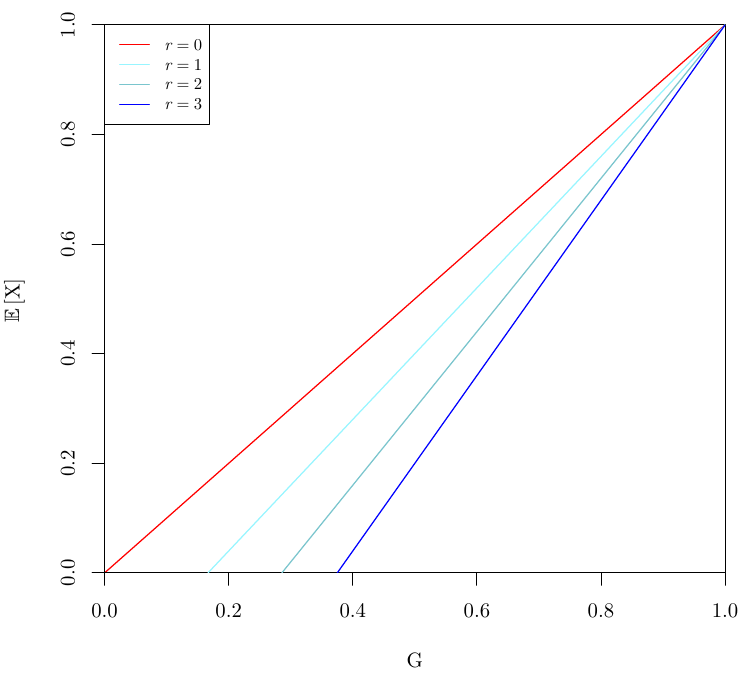}
    	\caption{The optimal forecast no longer equals the true expectation}\label{fig:Chloe_calibration}
\end{figure}

Clearly, there is an incentive to report a probability that is too high. This is not a general result. For example, had the forecaster been asked to forecast if Bob will win the series, i.e., had been asked to predict the complement, they would have been incentivized to report a probability that is lower than truly expected, as Bob can only lose early. When the time-varying event can both occur early and be ruled out early, the incentive to misreport exists, but it is not obvious in which direction forecasts will be influenced. For example, imagine we collected forecasts on whether a certain person will succeed the CEO of a company at the end of her/his term. This can happen at any time, but it can also be ruled out by a premature death of the candidate, a lifelong jail sentence, the collapse of the company or the appointment of another CEO.

\begin{rem}
The argument extends \textit{mutatis mutandis} to any number of time steps greater than two (and, by taking limits, to continuous time), as well as to any discounting scheme and number of potential outcomes. The only property required is that discounting induces \emph{outcome-dependent} weights, i.e., that the realized score is effectively multiplied by a factor that differs between at least two outcomes. In that case Lemma~4 of \citet{lindley_scoring_1982} applies and the resulting discounted score is (in general) no longer strictly proper.
In this case, \eqref{eq:dS^2_1}--\eqref{eq:dS^2_2}---and the optimal forecast $G$---change accordingly.
\end{rem}

\section{Predicting the Timing of Future Events}\label{sec:Continuous-time forecasts}

While Section~\ref{sec:Binary_forecasts} investigates probabilistic forecasts for the \textit{occurrence} of future events, we now investigate the incentive to misreport forecasts regarding the \textit{timing} of future events.
For technical ease, we treat time as continuous now. The event of interest, $X$, remains binary. The forecasting question thus is: 

\begin{center}
	When will an event $X$ happen during $[0,\tau)$?
\end{center}

Denote by $T\in[0,\tau]$ the random point in time when $X$ occurs.
Here, $T=\tau$ corresponds to the case where $X$ occurs at time $\tau$, or at some later time or, possibly even, never.
Denote the \textit{true} cumulative distribution function (cdf) of $T$ by $F(t)=\P\{T\leq t\}$.
Then, $F(\tau-)=\lim_{t\uparrow \tau}F(t)$ denotes the probability that $X$ occurs during $[0,\tau)$, and $1-F(\tau-)$ denotes the probability that $X$ occurs in $\tau$, afterward or never.
We denote a generic forecast for the cdf of $T$ by $G(\cdot)$.
A popular score to rank different distributional forecasts for $T$ is the continuous ranked probability score (CRPS), defined as
\[
	\CRPS(G,T) = \int_{-\infty}^{\infty}\big[G(t) - \1_{\{t\geq T\}}\big]^2\D t,
\]
where $G$ is a generic cdf forecast and $T$ is a verifying realization. The CRPS is essentially an extension of the quadratic scoring rule from \eqref{eq:(QSR)} to continuous variables; see \citet{MW76}, \citet{Her00} and \citet{GR07} for more details.
Therefore, this section extends the previous one.
The CRPS is \textit{strictly proper} (relative to the class of probability measures with finite first moments) in the sense that
\[
	\E\big[\CRPS(F,T)\big] < \E\big[\CRPS(G,T)\big]\quad\text{for all }G\neq F,
\]
where $F$ and $G$ possess finite first moments, and $T\sim F$ \citep[Sec.~4.2]{GR07}. 

In our setting with $T\in[0,\tau]$, the CRPS simplifies to
\[
	\CRPS(G,T) = \int_{0}^{\tau}\big[G(t) - \1_{\{t\geq T\}}\big]^2\D t.
\]
Note that once $X$ is realized, say at $T=t$, then the reward $\CRPS(G,t)$ for issuing forecast $G$ can be computed immediately, such that there is no need to wait until the terminal time $\tau$.
Therefore, the reward of $\CRPS(G,t)$ can be paid out at time $t$ (i.e., the point in time when $X$ occurs).
As in Section~\ref{sec:Binary_forecasts}, if this occurs, present-biased forecasters may be inclined to discount later payoffs (where $t$ is larger) more heavily than earlier ones.
Since present bias is one of the most robust features of human behavior, we model the payoff $\CRPS(G,t)$ as being discounted by a factor of $e^{-rt}$. The choice of the discount function is arbitrary and will be lifted later (see Remark~\ref{rem:discount_f}). 
Here, $r>0$ corresponds to a continuously compounded discount rate, with larger values of $r$ implying higher discounting of later rewards.
In light of these considerations, forecasters may not actually minimize the expected CRPS, but (implicitly or explicitly) minimize the expected \textit{discounted} CRPS
\begin{equation}\label{eq:dCRPS}
	\CRPS_d(G,t)=e^{-rt} \CRPS(G,t).
\end{equation}
When this happens, the optimal forecast no longer equals the true report, as we show next.

\begin{thm}\label{thm:main}
On the class of cdfs pertaining to distributions with compact support, the expected discounted CRPS, i.e., $\E\big[\CRPS_d(G,T)\big]$, is uniquely minimized by the cdf
\[
	G(t) = F(t)\frac{\E\big[e^{-rT}\mid T\leq t\big]}{\E\big[e^{-rT}\big]}.
\]
Moreover, $G(t)\geq F(t)$ for all $t\geq0$, such that it is rational to forecast earlier changes.
\end{thm}

\begin{proof}
The main idea of the proof of Theorem~\ref{thm:main} is to write $\E\big[\CRPS_d(G,T)\big]$ as an integral from $0$ to $\tau$.
Then, we apply the Euler--Lagrange method from the calculus of variations to find the minimizer of this integral.
See Section~\ref{S-sec:proof thm} of the Online Supplement for a detailed proof.
\end{proof}

\begin{rem}\label{rem:discount_f}
\begin{enumerate}[(i)]

\item The fact that the expected discounted CRPS is only minimized in the class of cdfs pertaining to distributions with compact support is not restrictive, because $T\in[0,\tau]$ by construction.

\item	We use a continuously compounded discount in \eqref{eq:dCRPS}. However, the exponential function $e^{-rt}$ may be replaced by any other discounting function $d(t)$, as long as it is non-increasing in $t$.
The proof of Theorem~\ref{thm:main} then goes through verbatim for the unique minimizer
\[
	G(t) = F(t)\frac{\E\big[d(T)\mid T\leq t\big]}{E\big[d(T)\big]}.
\]

\end{enumerate}

\end{rem}

\begin{rem}
Our setup, where a strictly proper score is premultiplied by a certain factor in \eqref{eq:dCRPS}, is reminiscent of the situation described by \citet[Eqn.~(2.9)]{Lea17} in their presentation of the forecaster's dilemma.
Adapted to our situation, they show that the expected weighted score $\E\big[w(t)\CRPS(G,t)\big]$ is minimized by the predictive density
\[
	g(t)=\frac{w(t)h(t)}{\int w(z)h(z)\D z},
\]
if $T$ has density $h$; see, in particular, the proof of Theorem~1 in \citet{GR11}.
This is consistent with our Theorem~\ref{thm:main}.
However, our Theorem~\ref{thm:main} is more general by not requiring $T$ to have a density. 
This is crucial, as $T$ has a natural point mass in $\tau$, capturing the probability that $X$ occurs only after time point $\tau$.
It is this generality that also leads to a much more involved proof compared with \citet[Proof of Theorem~1]{GR11}, where we need to draw on results from variational calculus.\footnote{The existence of a density of $T$ also allows \citet{GR11} to consider any strictly proper score (not just the CRPS) and any kind of weight function (instead of the exponential discount function).}
\end{rem}

\begin{rem}\label{rem:no discount}
It is instructive to consider the two extreme cases of no discounting (i.e., $r=0$) and increasingly heavier discounting (i.e., $r\to\infty$).
With $r=0$, Theorem~\ref{thm:main} implies
\[
	G(t) = F(t),
\]
which is the true cdf one would have expected from strict propriety of the (non-discounted) CRPS.

When $r\to\infty$, the situation is a bit more complicated.
Define $a=\essinf T$ to be the essential infimum of $T$.
Then, we show in Section~\ref{S-sec:proof eqn} of the Online Supplement that
\begin{equation}\label{eq:G r inf}
	\lim_{r\to\infty}G(t) = \begin{cases}0,& t<a,\\
	1,& t>a\\
	\1_{\{\P\{T=a\}>0\}}, & t=a.\end{cases}
\end{equation}
The see the intuition for this result, write
\[
	G(t)=\frac{\E\big[e^{-rT}\1_{\{T\leq t\}}\big]}{\E\big[e^{-rT}\big]}=\frac{\E\big[e^{-r(T-a)}\1_{\{T\leq t\}}\big]}{\E\big[e^{-r(T-a)}\big]}.
\]
As $r\to\infty$, the factor $e^{-r(T-a)}$ forces the expectation to concentrate on values of $T$ arbitrarily close to its essential infimum $a$.

To summarize the three cases when $r\to\infty$, we find that the optimal forecast $G(t)$ predicts an immediate occurrence of $X$ at the essential infimum of $T$.
That is, as soon as there is some chance that $X$ will occur, an ``infinitely impatient'' forecaster should place all probability mass at that point in time.
In particular, once it is clear that $X$ will not occur during $[0,\tau)$, then the true report $G=F$ minimizes the loss even for $r\to\infty$.
\end{rem}

\begin{figure}[t!]
	\centering
		\includegraphics[width=\textwidth]{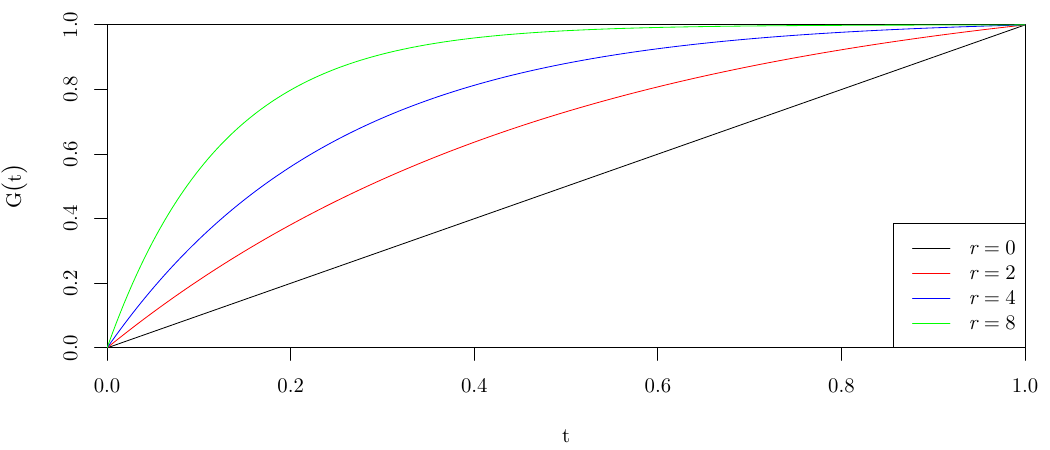}
	\caption{Plot of $G$ from \eqref{eq:Gr} for different values of $r$ and fixed $\tau=1$.}
		\label{fig:1}
\end{figure}

We now graphically illustrate how the optimal forecast shifts under discounting for various intermediate $r\in(0,\infty)$.
To do so, we assume that $T\sim\mathcal{U}[0,\tau]$, i.e., $T$ follows a uniform distribution supported on $[0,\tau]$.
The true cdf of $T$ is then given by $F(t)=t/\tau$ for $t\in[0,\tau]$ (see the black line in Figure~\ref{fig:1}).
It is easy to show that $\E[e^{-rT}\mid T\leq t]=\frac{1}{rt}(1-e^{-rt})$ in this case, such that
\begin{equation}\label{eq:Gr}
	G(t)=F(t) \frac{\frac{1}{rt}(1-e^{-rt})}{\frac{1}{r\tau}(1-e^{-r\tau})}=\frac{1-e^{-rt}}{1-e^{-r\tau}}.
\end{equation}
Therefore, under discounting, the optimal forecast is no longer a uniform distribution, but a truncated exponential distribution (with rate $r>0$).
Figure~\ref{fig:1} plots the above cdfs $G$ for different values of $r$.
In the case of no discounting (i.e., $r=0$), we have that $G=F$ (black line).
As $r$ increases, more probability mass is placed on earlier times, as the red, blue and green lines in Figure~\ref{fig:1} show.

\begin{rem}\label{rem:equiv}
Our above analysis shows that present-biased forecasters have an incentive to predict earlier occurrences of events.
Almost mechanically, this also translates into a higher predicted probability of the event occurring at all.
To see this, use Theorem~\ref{thm:main} and write
\begin{equation*}
	\E_{G}[X] = \P_{G}\{X=1\}= G(\tau-)\geq F(\tau-)=\P\{X=1\}=\E[X],
\end{equation*}
where $\P_{G}$ denotes the reported probability of the forecaster.
Therefore, we obtain a similar result as in Section~\ref{sec:Binary_forecasts} that the \textit{occurrence} of future events is reported to be more likely under present bias.
\end{rem}

\section{Empirical analysis}\label{sec:Empirical analysis}

\subsection{Preliminaries}

We now examine whether actual human forecasters alter their reports in response to the possibility of early payoffs.
In doing so, we build on the theoretical framework introduced in Section~\ref{sec:Binary_forecasts} for discrete time and extended to continuous time in Section~\ref{sec:Continuous-time forecasts} (cf.~Remark~\ref{rem:equiv}), where forecasts $G$ regarding the probability of future events $X$ are considered.
Here, we specifically analyze forecasts made on Metaculus---a reputation-based, massive online forecasting platform.
This allows us to identify the causal effect of time preferences---henceforth called \textit{time preference effects}---on reports, because the platform features both time-fixed events and time-varying events. 
Examples of time-fixed events from Metaculus are: 

\begin{itemize}
    \item At the Paris Summer Olympics in 2024, will the men's 100m dash winning time break the Olympic record of 9.63s?
    \item Will Republicans control the US Senate after the 2022 election?
    \item Will Apple announce a Virtual Reality headset at WWDC in 2023?
\end{itemize}
The verification dates of these questions are pre-specified. 
Examples of time-varying events from Metaculus are: 

\begin{itemize}
    \item Will Ukraine sever the land bridge between Crimea and Russia before 2024?
    \item Will Joe Biden announce before July 15, 2024 that he will not accept the Democratic Party's nomination for President?
    \item Will Emmanuel Macron dissolve the French National Assembly before the end of his term?
\end{itemize}
For these time-varying events it is not clear on which day outcomes can be verified.

At Metaculus, reputational tokens are awarded based on the logarithmic error of the forecasts. 
In principle, this is incentive-compatible, as the logarithmic error is strictly proper \citep{GR07}.
In practice, however, reputational tokens are awarded immediately after the outcome is determined by content moderators: for time-fixed events, this happens shortly after the scheduled date; for time-varying events, this happens soon after the event occurs, or after the scheduled end date if the event does not occur. Consequently, this introduces a time preference effect.\footnote{Additionally, Metaculus runs tournaments where there may be monetary prizes and ranks to be claimed. Most predictions are not monetarily incentivized. Any monetary reward would only materialize after a fixed time, so the monetary incentive cannot cause the behavior that we describe in the paper. Rather, the monetary reward should moderate it.}

Assessing distorted reporting behavior is a fundamentally difficult task because we can never observe true expectations.
However, if forecasters do respond to incentives to misreport due to time preferences, we should observe a systematic bias in forecasts for time-varying events that is not present in forecasts for time-fixed events. Any such systematic bias will be observable through the \textit{calibration} (sometimes also called \textit{empirical reliability}) which is simply how predictions correspond to actually observed frequencies of events \citep{parmigiani_decision_2009}. That is, ``a forecaster is well-calibrated if, for example, of those events to which he assigns a probability 30 percent, the long-run proportion that actually occurs turns out to be 30 percent'' \citep{dawid_well-calibrated_1982}. 
Forecasts on time-varying events should be systematically ``off the mark'',
as shown in Figure~\ref{fig:Chloe_calibration}.
Thus, keeping everything else constant, we interpret any systematic difference in calibration between predictions on time-fixed events and forecasts on time-varying events as attributable to time preference effects.

\subsection{Data: Metaculus forecasting tournament}\label{sec: Data}

To empirically examine incentives for misreporting, we analyze all binary events listed on the Metaculus platform at the time of data collection in August 2024 ($n=2005$). We denote whether the event occurred as $X \in \{0,1\}$; i.e., they could either occur or not. 775 events were classified by hand as time-fixed events, 798 as time-varying events that can only occur early, 197 as time-varying events that can only be ruled out early.
The remaining 235 events were classified as ambiguous, either because these events can both occur early and be ruled out early (the incentive to misreport is intangible) or due to ambiguity as to when the event would be considered resolved.\footnote{One example from our data includes the question ``Will Liverpool win the 2021--2022 Premier League?''. As Liverpool can theoretically secure winning the Premier League before the last game, or lose any chance of winning the Premier League before the last game, it is not clear which incentive forecasters may have to strategically misreport. This depends on how probable they deem either scenario, which can be heterogeneous.}
We only consider events that \textit{must have been verified} by August 2024, as including events with outcomes that could have been unobservable by August 2024 introduces a ``surprisingly early bias''.% \citep{lehmann_surprisingly-early_2025}.

\begin{table}[t!]
    \centering
    \begin{threeparttable}
    \caption{Statistics of the forecasting tournament's events}
    \begin{tabular}{ccc}
			\toprule
         & Time-fixed events & Time-varying events \\
         \hline
       count  & 775 & 798\\
       event occurred & 0.308 & 0.258 \\
       mean prediction  & 0.38 & 0.37 \\
       median prediction  & 0.35 & 0.32 \\
       mean number of predictions per event & 202.5 & 257.8\\
       median number of predictions per event & 106 & 161\\
			\bottomrule
    \end{tabular}
    \label{tab:check_for_balance}
    \begin{tablenotes}
      \small
      \item \textit{Note.} This table reports basic statistics for time-fixed events and time-varying events. Since the two essentially constitute treatment and control in our quasi-experimental setting, it is reassuring that the two groups are roughly similar in basic characteristics. 
    \end{tablenotes}
    \end{threeparttable}
\end{table}

Table~\ref{tab:check_for_balance} reports basic statistics for the two biggest groups of events. We see that (i) some events have received far more predictions than others, because the mean is far higher than the median, and that (ii) time-varying events received more predictions on average.
Apart from that, the sets of events are relatively similar.

We can gather suggestive visual evidence for the presence of misreporting by plotting the calibration curves for time-varying event forecasts and time-fixed event forecasts from our entire sample. 
In the analysis we only include time-varying events which can \textit{only occur early}, i.e., turn out to be 1, as opposed to events which can be ruled out early, i.e., turn out to be 0.\footnote{We add the latter group of events in a robustness check that is included in the Online Supplement. We find that the inclusion of this sample does not affect outcomes in a meaningful way and provides additional strong evidence for the presence of time preference effects.}
To estimate aggregate calibration, we take all forecasts that range from 1\% to 99\% assigned probability, and plot the association with the share of events that did occur, i.e., the frequency of the verified events, which results in Figure~\ref{fig:calibration}.

\begin{figure}[t!]
	\centering
		\includegraphics[width=\textwidth]{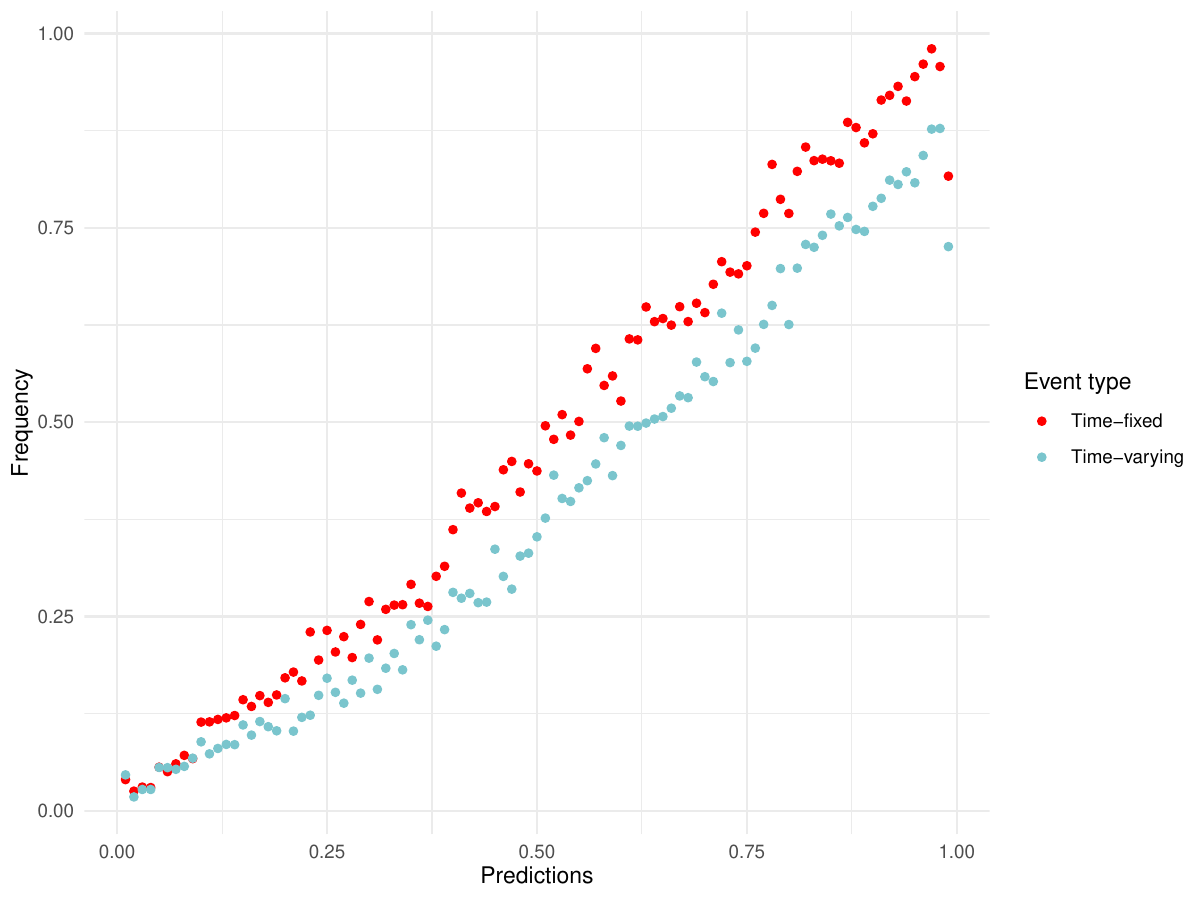}
        \caption{Metaculus forecasters' aggregate calibration}
        \vspace{1ex}
    \parbox{0.95\textwidth}{%
      \small
      \textit{Note.} This figure presents the calibration curve---often also called reliability curve---of 202 Metaculus forecasters meeting the criteria laid out in the main text. Each point assigns the reported forecasts the frequency with which actual events occurred. For example, of all 25\% probability forecasts, 23.2\% of the associated time-fixed events occurred, and 17.1\% of the associated time-varying events occurred.
    }
		\label{fig:calibration}
\end{figure}

 In Figure~\ref{fig:calibration} we see that time-varying events are systematically reported to be more probable conditional on their observed frequency. Interestingly, forecasts on time-fixed events are almost perfectly calibrated. This means that forecasts roughly correspond to the frequencies of time-fixed events, whereas forecasts on time-varying events are systematically too high, just as we would expect to see if the cause of this were time preference effects.

\begin{table}[t!]
    \centering
    \begin{threeparttable}
    \caption{Percentage share of predictions conditional on topic}
    \begin{tabular}{rrr}
		\toprule
                                & Time-fixed  & Time-varying\\
                                & $(n=94{,}707)$ & $(n=120{,}874)$\\
                                \midrule
        Politics                &  39.2 & 22.7 \\
        Elections               &  25.3 &  2.8 \\
        Economy \& Business     &   5.1 & 16.3 \\
        Geopolitics             &   3.5 & 14.6 \\
        Health \& Pandemics     &   2.5 & 11.9 \\
        Sports \& Entertainment &   9.0 &  2.5 \\
        Donald Trump            &   4.6 &  7.0 \\
        Computing \& Math       &   3.9 &  1.0 \\
        Technology              &   3.2 & 10.5 \\
        Ukraine                 &   1.5 &  3.1 \\
        Elon Musk               &   3.0 &  0.5 \\
        Natural Sciences        &   0.5 &  5.2 \\
        Russia                  &   1.0 & 10.3 \\
        Epidemiology            &   0.7 &  6.3 \\
        AI                      &   2.0 &  3.0 \\
        Virology                &   0.5 &  4.5 \\
        Other                   &  41.4 & 24.2 \\
			\bottomrule
    \end{tabular}
    \label{tab:predictions|topics}
    \begin{tablenotes}
      \small
      \item \textit{Note.} Each cell reports the fraction of events that fall into a given topic, conditional on the type of event. The two columns do not add up to 1 as an event can be assigned to multiple topics. Time-fixed questions are far more often about Politics and specifically Elections, whereas time-varying events are more often about Economy \& Business, Geopolitics, Health \& Pandemics and Technology. ``Other'' collects residual topics that are not explicitly listed here.
    \end{tablenotes}
    \end{threeparttable}
\end{table}

We cannot claim the calibration difference in Figure~\ref{fig:calibration} to be a causal effect of the incentives to misreport before addressing two potential sources of confounding. 
Firstly, there may be significant differences between the two sets of events---time-varying and time-fixed. 
Although the Metaculus forecasting tournament is full of rich and diverse events, some events are by their very nature more often time-varying, such as events related to technological progress (``When will a SpaceX Starship reach orbit?''). Other events are by their very nature fixed in time, such as elections. 
Events on the Metaculus platform are automatically assigned a number of topics.\footnote{See: \texttt{https://www.metaculus.com/notebooks/21576/streamlined-llm-enhanced-question-discovery/}} Table~\ref{tab:predictions|topics} reports how predictions are distributed conditional on topics and event type. We observe that there are great differences between the two sets of predictions. Table~\ref{tab:predictions|topics} includes only the most frequent topics. Predictions on topics that are not included in Table~\ref{tab:predictions|topics} are labeled ``Other''. The topics are not exclusive, i.e., events labeled ``Politics'' can also be labeled ``Elections''.
Now, if forecasters consistently misjudge certain types of events---such as being overly optimistic about technological progress---then this bias may be associated with the event classification as time-varying or time-fixed, potentially confounding the effect of interest. To address differences in topics, we balance our sample so that topics are represented equally, as described in Section~\ref{sec: IPTW}.

Furthermore, we must acknowledge that forecasters are not forced to make forecasts on any particular event.
Therefore, the second source of confounding is that time-fixed events and time-varying events may attract systematically different forecasters, potentially causing differences in calibration. We address this by using a within-subject design, eliminating this potential source of confounding. This is described in Section~\ref{sec: Regression analysis}.

To preview the main results, we find that the aggregate picture in Figure~\ref{fig:calibration} remains intact, and that time preference effects are a powerful driver of rationally dishonest reporting,
as none of the added adjustments for (i) topics and (ii) forecaster self-selection meaningfully affects the results.

Finally, one might be inclined to view question difficulty or ``inherent unpredictability'' to be an additional confounder. After all, forecasts on time-fixed questions may be objectively easier or harder than forecasts on time-varying events. Whilst that may very well be true, differences in inherent unpredictability should not confound our analysis because we are studying the calibration, not the \textit{precision} or sharpness of forecasts. Fundamentally, strictly proper scoring rules---such as the quadratic scoring rule---can be decomposed into error from systematic bias (calibration) and ``how spread out the forecasts are'' (precision/sharpness) \citep{brocker_reliability_2009, parmigiani_decision_2009}. Changes in inherent unpredictability should affect precision, not the calibration.\footnote{
We can use archery as an analogy to explain the problem. An archer's accuracy is judged by how close their arrows land to the center of the target. Ideally, the archer should not systematically miss the target. The average landing spot of the arrows should be near the center. If the arrows tend to gather away from the center, the archer is miscalibrated, meaning their shots are consistently off target.
How difficult it is to hit the target depends on how far the archer is from it. The farther away they stand, the harder it is, and the more spread out the shots will be. We assume that being further away doesn't mean the archer becomes miscalibrated; the shots are just less accurate, not consistently off-center. Similarly, we assume that forecasters calibration does not change with question difficulty, although the precision of forecasts will be affected.
As we will see later, the data do support the assumption that calibration is independent of question difficulty. If question difficulty was indeed a strong confounder, impacting calibration, we would see a drastic change in calibration as a result of controlling for topics, which should mix up question difficulty. Since re-balancing the forecasting data based on topics does not seem to affect calibration to a large degree, we find it highly plausible that both topics and difficulty (which cannot be disentangled here) do not meaningfully affect calibration.}

\subsection{Inverse probability of treatment weighting}\label{sec: IPTW}

We address the concern that differences between forecasts are driven by systematic differences in topics through weighting forecasts by topics in order to create an artificially balanced sample, which is a way of controlling for the differences in topics.
This method---called \textit{inverse probability of treatment weighting} (IPTW)---effectively gives more weight to forecasts that are on likely-to-be time-varying topics (Geopolitics, Economy, ...) in the control group (time-fixed events) and more weight to forecasts that are on likely-to-be time-fixed topics (Elections, Sports, ...) in the treatment group (time-varying events), thus automatically balancing the sample \citep{austin_moving_2015}. 
We refer to Section~\ref{S-sec:Additional detail related to data and sample construction} of the Online Supplement for details. In other words, we create within-forecaster samples of time-varying event and time-fixed event predictions that are---on average---more comparable in terms of topics, if we weight the forecasts carefully.
The weights are calculated by using the probability to be treated, i.e., the probability that a certain forecast is on a time-varying event: 
\begin{align*}
    \text{Weight}_\text{tvarying}  & = \frac{1}{\P(tvarying=1 \mid \mathbf{topics})}, \\
    \text{Weight}_\text{tfixed}& =  \frac{1}{1 - \P(tvarying=1 \mid \mathbf{topics})},
\end{align*}
where $tvarying$ is an indicator variable that equals $1$ for a time-varying event and $0$ otherwise.
We estimate the probability $\P(tvarying=1 \mid \mathbf{topics})$ via the standard logistic regression
\begin{equation*}
\P(tvarying=1 \mid \mathbf{topics}) = \frac{1}{1 + e^{- \left( \beta_0 + \bm{\beta}^\top \mathbf{topics} \right)}},
\end{equation*}
where $\mathbf{topics}$ is a vector of 80 observable covariates (e.g., Politics).\footnote{Details regarding the covariates and coefficients can be found in the Online Supplement (Section~\ref{S-Additional details related to the Inverse Probability of Treatment Weighting}).}

IPTW is a variation of propensity score matching, which has been found to sometimes insufficiently address bias between treatment and control samples, largely because of lacking or improper control variables \citep{smith_reconciling_2001}. There is no reason to believe that we are faced with such an issue in this study 
because there is no complex causal interdependency between forecasts and the events that they are referring to---quite unlike most observational studies where propensity score matching is used. Therefore, controlling for additional covariates/topics is unlikely to be causing systematically wrong estimates. Furthermore, we have a very rich set of covariates/topics that we control for and we do not engage in propensity score \textit{matching} but simply balance our entire sample on topics. 
We do not find evidence that the weighting of forecasts meaningfully affects the outcome of the analysis. This result is discussed in the next section.

\subsection{Regression analysis}\label{sec: Regression analysis}

Another cause for the differences in calibration in Figure~\ref{fig:calibration} could be that forecasters who issue more time-varying event forecasts are systematically different from those that make more time-fixed event forecasts.
In this case, differences in aggregate calibration between the two groups of questions would actually reflect differences between individual forecasters. 

We can eliminate this concern by looking at the calibration within subjects, effectively controlling for a potential self-selection of forecasters to events. 
By using a within-subject design, we ask: What is the expected difference between forecasts on time-varying events and time-fixed events of an individual forecaster? We use the balanced sample created by IPTW (see Section~\ref{sec: IPTW}).

We collect forecasts from each forecaster $i$, which yields $k=202$ panels of forecasts. 
The data structure is illustrated in Table~\ref{tab:example_data_structure}. We denote a probabilistic forecast by $G\in[0,1]$. We then measure the differences in calibration between time-fixed event forecasts and time-varying event forecasts \textit{within} individual forecasters and combine these differences \textit{across} forecasters in a random-effects meta-analysis using the \texttt{R} package \texttt{metafor}'s function \texttt{rma()} \citep{borenstein_basic_2010,viechtbauer_conducting_2010}.
Thus, we arrive at an average ``incentive treatment'' effect and can test whether it is significantly different from zero. 

\begin{table}[t!]
    \centering
    \begin{threeparttable}
    \caption{Data structure}\label{tab:example_data_structure}
            \begin{tabular}{ccccccc}
                \toprule
                Prediction ID & Forecaster ID & Forecast $G$ & Outcome $X$ & IPTW & tvarying & [topics] \\
                \midrule
                1  & 10844 & 0.24  & 0  & 1.133 & 1 & ... \\
                2  & 10844 & 0.3  & 1  & 1.495 & 1 & ... \\
                3  & 10844 & 0.69  & 1  & 1.982 & 0 & ... \\
                4  & 10844 & 0.59  & 1  & 1.815 & 0 & ... \\
								$\vdots$  & $\vdots$ & $\vdots$  & $\vdots$  & $\vdots$ & $\vdots$ & $\ddots$ \\
                \bottomrule
            \end{tabular}
    \begin{tablenotes}
      \small
      \item \textit{Note.} We have 202 panel members (number of forecasters eligible for the study) of the kind depicted in this table. Each panel contains all predictions made by a single forecaster. The variable 'IPTW' is the weight calculated from propensity scores. 
    \end{tablenotes}
    \end{threeparttable}
\end{table}

For each forecaster $i$, we estimate a calibration curve $\mu_{i}(G) := \P(X=1 \mid i,G)$, i.e., $\mu_{i}$ is forecaster $i$'s observed event rate among predictions of value $G$. We obtain $\mu_{i}$ by binning forecaster $i$'s predictions at the 0.01-level. We estimate the calibration $\mu_{i}$ as

\begin{equation}\label{eq:regression}
 \mu_{i} = \underbrace{\beta_0^{(i)}+\beta_1^{(i)} \ G}_{\text{time-fixed events}}  + \underbrace{\beta_2^{(i)} \ tvarying + \beta_3^{(i)} \ tvarying \times G}_{\text{time-varying events}}  + \varepsilon_{i},\qquad\E[\varepsilon_{i}\mid G,\ tvarying]=0,
\end{equation}

where $\beta_0$ denotes the intercept, $\beta_1$ the slope for the time-fixed event calibration, and $\beta_2$ and $\beta_3$ measure the change in intercept and slope for time-varying events.\footnote{As the data may be heteroskedastic, we use heteroskedasticity-robust standard errors from \citet{mackinnon_heteroskedasticity-consistent_1985}.}
Although calibration must not be linear \textit{per se}, perfect calibration, where reported probabilities correspond to observed frequencies, is linear; see also Figure~\ref{fig:Chloe_calibration}.\footnote{The calibration of misreported predictions in Section~\ref{sec:Binary_forecasts} is also linear. However, this depends on an arguably arbitrary modeling choice.} Thus, we use linear regression as the most valid available approach.
We practically estimate two linear functions---one for time-varying event forecasts and one for time-fixed event forecasts. However, by using time-varying events as a treatment dummy we incorporate both into one model, as in \eqref{eq:regression}.

We test that individual forecasters are \textit{not} systematically misreporting their expectations of time-varying events. 
This is equivalent to the variable $tvarying$, that signals whether the event is time-varying, having no predictive power, i.e.,
\begin{align*}
    \mathcal{H}_0^{tvarying}\colon\quad \beta_2=0,\ \beta_3 = 0.
\end{align*}

To a lesser extent, we furthermore expect forecasters to be well-calibrated, i.e., their expectations to be accurate. Since this is in line with an intercept of 0 and a slope of 1, 
\begin{align*}
    \mathcal{H}_0^{calib}\colon\quad \beta_0=0,\ \beta_1 = 1.
\end{align*}

In order to limit the influence that a single event might have on the calibration of individual forecasters we restrict our analysis to forecasters who have made predictions on at least 40 \textit{different} time-fixed and time-varying events, and at least a total of 80 predictions on time-varying and time-fixed events respectively.
Forecasters ($k=202$) can make multiple predictions on the same event, which is why we require both a sufficient number of predictions and a minimum number of events that these predictions refer to. Multiple predictions on the same event are common, as they reflect new information (updates). Within-event dependence arises as the event \textit{outcome} on updated predictions will be the same. Furthermore, the updated predictions could be serially correlated. We run a robustness check using only one forecast per event per forecaster (first or last).
It turns out that this is not an issue, as collapsing the data in this way does not meaningfully affect our result, and refer to the supplemental materials for details.
We obtain the mean estimates and standard errors of the four coefficients $\beta_0,\ldots,\beta_3$ using a standard random-effects model, and test the joint hypotheses $\mathcal{H}_0^{tvarying}$ and $\mathcal{H}_0^{calib}$ using the Bonferroni correction \citep{miller_simultaneous_1981}.\footnote{We use the \texttt{metafor} package in \texttt{R} and estimates of heterogeneity across forecasters as specified in \citet{dersimonian_meta-analysis_1986}. Graphical representations in the form of forestplots of the effects (Figures~\ref{S-FigForest1}--\ref{S-FigForest4}) can be found in the Online Supplement (Section~\ref{S-sec: Detail_regression}).}

\begin{table}[t!]
    \centering
    \begin{threeparttable}
    \caption{Regression Estimates}\label{tab:regrssn_outcomes}
    \begin{tabular}{lccc}
        \toprule
        & \textbf{(1)} &        \textbf{(2)}                    & \textbf{(3)}\\
        & \textbf{Aggregate} & \textbf{Within-subject} & \textbf{Within-subject \& IPTW} \\
        \midrule
         $\beta_0$  & $-0.0517$*** & $-0.0514$*** & $-0.0507$*** \\
                       & (0.0087) & (0.0055) & (0.0055) 
                       \vspace{10pt}\\
            $\beta_1$  & $1.0558$*** & $1.0671$*** & $1.0665$*** \\
                       & (0.0151) & (0.0106) & (0.0108) 
                       \vspace{10pt}\\
            $\beta_2$  & $-0.0279$* & $0.0101$* & $0.0095$ \\
                       & (0.0123) & (0.0048) & (0.0049) 
                       \vspace{10pt}\\
            $\beta_3$  & $-0.0948$*** & $-0.1309$*** & $-0.129$*** \\
                       & (0.0214) & (0.0122) & (0.0124) \\
        \bottomrule
    \end{tabular}
    \begin{tablenotes}[flushleft]
      \small
      \item \textit{Note:} * $p\leq0.05$; ** $p<0.01$; *** $p<0.001$. $p$-values correspond to individual hypothesis tests against zero for $\beta_0$,$\beta_2$,$\beta_3$, and against one for $\beta_1$.
      \item $\beta_0$ corresponds to the intercept. 
            $\beta_1$ ($G$) corresponds to the slope of the calibration curve for time-fixed event predictions. 
            $\beta_2$ ($tvarying$) corresponds to the \textit{change} in intercept for time-varying event predictions. $\beta_3$ ($tvarying \times G$) to the \textit{change} in slope of the calibration curve for time-varying event predictions.
    \end{tablenotes}
    \end{threeparttable}
\end{table}

Doing so, we reject both hypotheses $\mathcal{H}_0^{tvarying}$ and $\mathcal{H}_0^{calib}$ since $\beta_0$,$\beta_1$, and $\beta_3$ are significantly different from the hypothesized value; indeed, the Bonferroni-adjusted $p$-values are smaller than $0.001$ for both hypotheses. The detailed outcomes are reported in Table~\ref{tab:regrssn_outcomes}, where column 3 corresponds to the regression including controls for topics (IPTW) and forecaster self-selection (Within-subject).
To see how the added controls affect the result, we run a simple OLS regression on the aggregate predictions dataset, dropping both the within-subject control and topic balancing. This regression---which is essentially a fitted line in Figure~\ref{fig:calibration}---is in line with both theoretical expectations and the more carefully controlled estimates. We report estimated coefficients in column 1 of Table~\ref{tab:regrssn_outcomes}. Additionally, we repeat the regression analysis using only the within-subject specification but without balancing for topics. These results are reported in the second column in Table~\ref{tab:regrssn_outcomes}.

Forecasters are not perfectly calibrated 
because the intercept $\beta_0$ is smaller than 0 and the slope for $\beta_1$ is larger than 1. This suggests that forecasters skew their predictions toward 50\%, such that the calibration curve is slightly ``s-shaped''.\footnote{Such a ``center bias'' is commonly observed in forecasting data and is partially explainable by risk aversion \citep{hurley_experimental_2005,danz_belief_2022}.}
Nonetheless, the calibration of forecasters on time-fixed events is very good overall.

On the other hand, for time-varying events, the calibration of forecasters is significantly different, i.e., time preferences affect predictions to a significant degree. 
Forecasters report systematically higher predictions when the event is time-varying, as the difference in slope is negative ($\beta_3 <0$). A lower coefficient refers to a lower frequency \textit{ceteris paribus}. Thus, a lower coefficient implies a higher prediction for any given frequency. Our analysis finds that $\beta_2 > 0$ (although the effect is not statistically significant), which means that the intercept difference between control and treatment is positive (though small).

We can combine these observations by looking at all coefficients and (model-implied) calibration in Table~\ref{tab:Expected_prediction_frequency}, which reports the expected prediction $\E[G \mid tvarying, \mu]$ implied by \eqref{eq:regression} for a given true event probability $\mu \in (0,1)$, for time-fixed and time-varying events, i.e., the regression inverted to solve for $G$. Predictions on time-varying events \textit{are} higher than predictions on time-fixed events because the difference in slope $\beta_3$ is far larger than---and completely dominates---$\beta_2$, which is close to zero. Why is $\beta_2$ still positive? We conjecture that $\beta_2$ is positive because $\beta_0$ is already large and negative for reasons unrelated to time preferences, such as all predictions involving low-probability judgment being too high. Thus, $\beta_0$ might be masking potential differences between intercepts ($\beta_2$).

We see that the extent to which forecasters misreport their own expectation is quite large.
Table~\ref{tab:Expected_prediction_frequency} shows that the expected prediction from an average forecaster for a time-fixed event which occurs with 30\% probability would be 32.9\%, whilst the expected prediction for an equally likely time-varying event would be 36.4\%. On time-varying events that are expected to occur with 70\% probability (frequency and fixed-event-prediction), forecasters reported a whopping 79\% probability on average. Table~\ref{tab:Expected_prediction_frequency} reports that predictions on time-varying events that occurred with 90\% frequency are at an impossible 100.4\%. From Figure~\ref{fig:calibration} we can see that the observed frequency of time-varying events is never above 90\%, such that it would indeed take a larger-than-100\% prediction to get to a frequency of over 90\% for time-varying events.

\begin{table}[t!]
    \centering
    \caption{Expected prediction given observed frequencies}
    \begin{tabular}{cccccccccc}
			\toprule
       $\mu$  & 0.1 & 0.2 & 0.3 & 0.4 & 0.5 & 0.6 & 0.7 & 0.8 & 0.9 \\
       \midrule
       $\E[G \mid tvarying=0, \mu]$  & 0.141 & 0.235 & 0.329 & 0.422 & 0.516 & 0.610 & 0.704 & 0.798 & 0.891\\
       $\E[G \mid tvarying=1, \mu]$ & 0.151 & 0.257 & 0.364 & 0.471 & 0.577 & 0.684 & 0.791 & 0.897 & 1.004 \\
			\bottomrule
    \end{tabular}
    \begin{tablenotes}[center]
      \small
      \item \textit{Note:} Expected predictions implied by fitted model from \eqref{eq:regression}.
    \end{tablenotes}
    \label{tab:Expected_prediction_frequency}
\end{table}

We find no pattern in how much forecasters misreport their predictions. Moderate misreporting seems to be common across forecasters and not strongly correlated with how many predictions a forecaster made. 
We plot the estimates for $\beta_3$ for the most active 39 forecasters in the forestplot in Figure~\ref{fig:top_of_forestplot}. The farther the estimates of $\beta_3$ are from the center, the stronger the difference in slope between calibration on time-fixed and time-varying events. 
The individual forecasters are sorted from most predictions made (top) to least predictions made. The bottom entry marks the average estimate.

\begin{figure}
    \centering
    \includegraphics[width=0.6\textwidth]{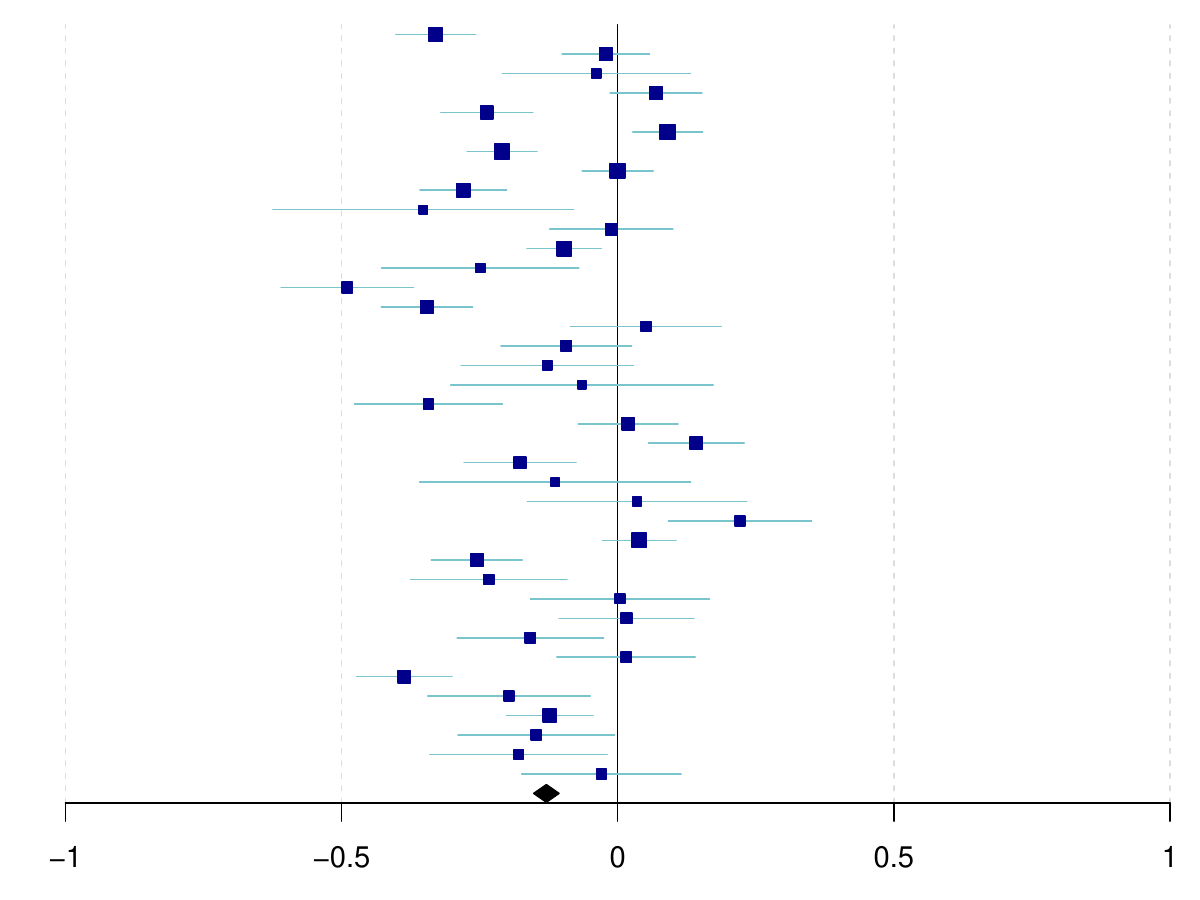}
    \caption{Estimates of $\beta_3$ for the 39 most active forecasters. }\label{fig:top_of_forestplot}
            \vspace{1ex}
    \parbox{0.65\textwidth}{%
      \small
      \textit{Note.} This figure presents the forestplot of 39 individual within-forecaster estimates of the coefficient $\beta_3$. The forecasters are sorted (top to bottom) from most predictions made (total) to least predictions made. The horizontal axis refers to the coefficient's size. Whiskers indicate the 95\%-confidence intervals. The mean estimate for all forecasters is at the bottom.
    }
\end{figure}

Finally, we are interested in how time preferences impact the accuracy of forecasts. The effect on accuracy is negative
because a strictly worse calibration will \textit{ceteris paribus} lead to lower accuracy. Since any strictly proper scoring rule penalizes systematic bias (calibration) and precision (how close forecasts are to true outcomes) separately \citep{GR07}, we can isolate the 
contribution of systematic bias to the accuracy, as measured by the mean quadratic scoring rule or Brier score in our dataset. We find that the predictions on time-fixed events have an average systematic bias that contributes $0.0023$ to the Brier score, whereas the predictions on time-varying events have an average systematic bias that contributes roughly $0.011$ to the Brier score. Setting the calibration of time-varying-event predictions equal to those of time-fixed-event predictions would have improved accuracy by more than 3.5\%.\footnote{We refer to Section~\ref{S-sec:Effect_on_accuracy} of the Online Supplement for more details.} This is just another way of saying that time-varying event predictions are much less well calibrated.

\section{Discussion}\label{sec:Discussion}

This study provides strong evidence that time preferences impact real-world forecasts.
We bring forward intuitive, theoretical, and empirical accounts, all of which provide evidence of their own. 
Whilst our empirical study is observational---limiting our ability to assure equality of treatment and control---the study has the huge advantage of being conducted in the real world with actual forecasts on future events spanning multiple years. Therefore, our empirical study possesses strong external validity. Furthermore, we find time preference effects in all model specifications; see Table~\ref{tab:regrssn_outcomes}. This leads us to believe that differences between treatment and control, which are minimized using our IPTW and the within-subject design, do not strongly confound the measured time preference effects. 

We remark that additional factors need to be considered when interpreting the results. 
A cognitive bias that is closely associated with time-varying events cannot be distinguished from an effect that is caused by human preferences.\footnote{We thank an attendant of our talks for this remark.}
However, the data suggests that the observed effect is caused by time preferences.  
The reason for this lies in the rule of complement: We can take any event and ask for the complement, i.e., the chance that it will not happen. Theoretically, as the predictions on events that can happen early are inflated, the complement would have to be deflated, i.e., predictions should be systematically too low. We indeed see this pattern in a limited set of events on Metaculus where the questions are formulated to ask for the complement.\footnote{We refer to Figure~\ref{S-fig:calibration_dynamic_-} in the Online Supplement. We also discuss how follow-up studies could collect evidence on either hypothesis.} It seems implausible that a cognitive bias would cause distortions that are sensitive to the rule of complements.

In fact, we suspect that our empirical study far underestimates the degree to which predictions would be affected had forecasters minimized their expected discounted error.
Firstly, we know from the literature on belief elicitation that humans have a strong preference for being seen as honest and actually being honest \citep{abeler_preferences_2019}. Therefore, humans seem to strike a balance between honesty and reward-maximization when they can gain from misreporting their beliefs. Forecasters may be distorting their forecasts slightly but not drastically, even if they would benefit from doing so.
Secondly, forecasters may be (partially) unaware of the scoring and the opportunity to gain from misreporting.
Furthermore, our sample contains mainly short-term forecasts. Although a handful of forecasts is as much as 7-years-ahead, the vast majority of forecasts is less than 1-year- ahead (median: $118$ days, mean: $227$ days). Since the issue of time preferences would increase with the expected time-to-event, we expect that long-term forecasts suffer more heavily from time preference effects than short-term forecasts. Most long-term forecasts from Metaculus are not in our dataset because the events are not yet verified.

A natural question arises: How large is this distortion in practice? The magnitude of the distortion should depend on both the discount rate $r$ and the expected time-to-event. The longer the time-horizon over which events are expected to unfold, the larger the distortion should be. However, estimating the size of the distortion as a function of time-to-event (or estimating $r$) is less straightforward than it may seem. In our empirical study, simply conditioning on the observed time-to-event introduces a surprisingly-early bias, a form of selection bias: events with a very short observed time-to-event are, almost by definition, events that occurred sooner than expected. In such cases, we should not expect observed frequencies to match forecasted probabilities, even if those forecasts were perfectly calibrated.

Overall, we find that time preference effects are a serious problem that can systematically affect forecasts in all areas, and that forecasting practitioners should be aware of. 
Moreover, this problem extends towards all statements regarding the future which can be ``right'' earlier than they can be ``wrong''.
Failures and disruptions are more likely to occur---and become apparent---in the short-term than sustained, smooth operations or simply ``business as usual''. Therefore, time preferences should lead to overestimation of risks in various areas, such as
credit risk, life expectancy estimation, natural disaster forecasting, reliability engineering, technological progress forecasting, pandemic forecasts, geopolitical analysis, environmental risk, and information security.\footnote{It is commonly understood that experts may overestimate risks in order to err on the side of caution. This is often (mistakenly) called risk aversion or conservatism. However, \textit{if} we wanted to err on the side of caution, being systematically misinformed about risks is not necessary. Risk-averse actors can simply \textit{choose} to take the safe route based on unbiased estimates of risk. It seems far more plausible that experts overestimate risks because it is better for \textit{them}, not only because of time preferences, but also because unexpected failures are certainly penalized more harshly than an undue overestimate of that risk. Thus, experts may overestimate risks in order to protect their reputation, or avoid blame.}
As a result, we should expect that decision-makers and regulators apply standards in systematically suboptimal ways, over- or underemphasizing safety and preparedness. In particular, forecasts for technological milestones seem to be overly optimistic \citep{tichy_over-optimism_2004}, which could---at least partially---be explained by time preference effects.

As \citet{ottaviani_strategy_2006} put it: ``Forecasting is proving to be an apt laboratory for improving our understanding of
strategic communication and positioning by non-partisan informed agents. The availability
of data sets and the richness of institutional details can inspire and give discipline to our
theorizing. The insights gained can be helpful in shedding light on a number of other social
and economic problems [...]''.
We remark that time preferences do not necessarily lead to a decrease in the accuracy of forecasts---as they do in the Metaculus tournament---if they happen to cancel out or mitigate other existing biases. 
However, in most cases we expect time preferences to be unwanted and potentially harmful. 

The remaining question is how to reduce or eliminate them.
We identify four different strategies to eliminate or cope with time preference effects and investigate each in turn. However, there is no solution that comes without drawbacks.

\begin{enumerate}
    \item Increase future rewards proportionally to offset discounting

We can offset the discounting of future rewards by increasing rewards proportionally over time, i.e., multiply \eqref{eq:dCRPS} by $e^{rt}$.
However, this may be difficult in practice because it (i) requires knowledge of the forecasters discount rate $r$ and (ii) greatly complicates the nominal calculation of payoffs or reputational tokens to the forecaster.

    \item Substitute time-varying events with time-fixed events

An obvious fix would be to just wait until some defined threshold of time $\tau$ when evaluating, and---more importantly---rewarding forecasts.
Instead of asking ``When will China invade Taiwan?'' we can ask ``Will China have invaded Taiwan on the 1st of January?'' \textit{and} wait with the reward until the 1st of January.
There are multiple problems with this approach, the first is determining $\tau$. For many areas of interest, e.g., earthquakes and climate change, we do not want $\tau$ to be so far away in the future that forecasters do not care about evaluation. However, we would like to be able to elicit long-term predictions. We could forecast multiple of such time-fixed events, varying $\tau$ each time. This approach is equivalent to approximating a probability density distribution (that is, the forecast $G$ on the time-varying event) by a discrete histogram (which is composed of the forecasts on time-fixed events). In the limit, with infinitely many time-fixed event forecasts (one corresponding to each point in time), the two will converge. This would be equivalent to continuously evaluating a forecast.
Undoubtedly, continuously evaluating forecasts is difficult in practice. The approach of substituting time-varying events with time-fixed events eliminates time preference effects because evaluation is occurring in one point in time only. However, it also discards valuable information as a probability density distribution contains more information than a histogram.
A second problem is that forecast evaluation is not always formal. For example, imagine a public intellectual sharing a prediction on social media, or a manager sharing a prediction with their peer group. If that prediction turns out to have been accurate in the short-term, it will be perceived as insightful, immediately boosting their reputation, although no formal scoring occurs. 
Notwithstanding its drawbacks, we believe that this approach---avoiding time preferences altogether---is the best option for most areas of forecasting because of its simplicity.
    
    \item Offer no rewards to forecasters or do not assign an error to forecasts 

Since incentives are the culprit that lead to time preference effects, we can eliminate time preference effects by eliminating incentives related to forecasts.
However, this is both undesirable and difficult.
This approach is undesirable because we would like to evaluate forecasts to ensure that they are accurate. It is difficult to eliminate incentives because incentives are not only payments made to the forecaster. In our empirical study, reputational incentives---which are less tangible and harder to eliminate---are sufficient to give rise to strong time preference effects.

    \item Correct predictions \textit{ex post} by the estimated time preference effect 

Given that it is easy to identify time-varying event forecasts, and all of them are potentially subject to time preference effects, we can try to reverse-engineer the actual expectation from the reported forecast. However, this requires knowledge of exactly how forecasters misreport.

\end{enumerate}

This enumeration of possible solutions suggests
that forecasts of time-varying events are inherently challenging to evaluate and incentivize. Any attempt to do so risks creating misaligned incentives.
Given the importance of foresight related to future events, this is an important finding of its own. 

The reason for its unsolvability is obvious once we understand time preference effects (as defined in this study) as the less-well-known cousin of the general problem of long-term forecasting: it is difficult to incentivize long-term predictions given that their evaluation is far away.
Since there is no easy ``solution'' to this problem, it follows that time-varying event predictions---even if relatively short-term---are also affected.

\section{Conclusion}\label{sec:Conclusion}

This study investigates the incentive to report earlier occurrence of events in order to increase (perceived) short-term benefits.
If a forecasted event can occur earlier than other potential outcomes, forecasters can ``bet'' on this early resolution and can only be ``correct'' in the short term. As a result of such distortions, overall forecasting performance decreases. Utilizing forecasting data from Metaculus, we observe a substantial and significant inflation of predictions when the forecasted event can occur early.
Therefore, we have reason to believe that forecasts in important domains such as technological forecasting, extreme weather forecasting, pandemic forecasting, geopolitical and financial risk may be systematically biased. We also find that there is limited practical leeway for aligning incentives with honest reporting, such that awareness related to potentially misaligned incentives is critically important.
Future research in this area could involve studying time preference effects in applied areas. 
Furthermore, conducting a controlled experiment could deepen our understanding and provide valuable evidence in addition to this study. Such an experiment would allow us to obtain various auxiliary measurements, including the implied discount factor $r$.

\singlespacing
\small

\section*{Disclosure Statement}

The authors report there are no competing interests to declare.

\section*{Data Availability Statement}

The code used to generate the results and figures in this paper is publicly available in a GitHub repository at \nolinkurl{https://github.com/PfadQualle/time_preferences_in_forecasting}. This repository includes all scripts necessary to reproduce the analysis, conditional on access to the underlying data.
The forecasting data used in this study were obtained from Metaculus and are subject to restrictions that prevent redistribution. Researchers seeking access to these data should contact Metaculus directly. 
Upon reasonable request, and where permitted by the terms of our data use agreement with Metaculus, the corresponding author will be happy to provide guidance on replicating and extending the empirical analysis in this study.

\bibliographystyle{apalike}
\bibliography{thebib}

\end{document}

\baselineskip18pt
\renewcommand\floatpagefraction{.9}
\renewcommand\topfraction{.9}
\renewcommand\bottomfraction{.9}
\renewcommand\textfraction{.1}
\setcounter{totalnumber}{50}
\setcounter{topnumber}{50}
\setcounter{bottomnumber}{50}
\abovedisplayskip1.5ex plus1ex minus1ex
\belowdisplayskip1.5ex plus1ex minus1ex
\abovedisplayshortskip1.5ex plus1ex minus1ex
\belowdisplayshortskip1.5ex plus1ex minus1ex

\title{
Supplementary materials to ``Time preference effects in forecasting''
}

\author{
	Yannick Hoga\thanks{Faculty of Economics and Business Administration, University of Duisburg-Essen, Universit\"atsstra\ss e 12, D--45117 Essen, Germany, \href{mailto:yannick.hoga@vwl.uni-due.de}{yannick.hoga@vwl.uni-due.de}.}
	\and 
		Niklas V. Lehmann\thanks{Faculty of Business Administration, TU Bergakademie Freiberg, Schlossplatz 1, D--09599 Freiberg, Germany, \href{mailto:Niklas-Valentin.Lehmann@vwl.tu-freiberg.de}{Niklas-Valentin.Lehmann@vwl.tu-freiberg.de}.}
}

\date{\today}
\maketitle

\begin{abstract}
	\noindent 
	Throughout, references starting with an ``S.'' refer to this supplement, whereas references without this prefix refer to the main paper.
	This supplement contains Section~\ref{sec:Proofs}, which contains the proofs of our technical results.
	Then, Section~\ref{sec:Methodology} provides additional details on the empirical analysis.
	Section~\ref{sec:Robustness checks} carries out several robustness checks.
	In Section~\ref{sec:Effect_on_accuracy}, we calculate the effect of time preferences on accuracy.
	Finally, Section~\ref{sec:Forestplots for the Random-Effects Meta Analysis} collects forestplots for the random-effects meta analysis.
\end{abstract}

\doublespacing

\section{Proofs}\label{sec:Proofs}

\subsection{Proof of Theorem~\ref{M-thm:main}}\label{sec:proof thm}

\begin{proof}[\textbf{Proof of Theorem~\ref{M-thm:main}:}]
The expected discounted CRPS for some issued distributional forecast $G$ is, by the law of total probability,
\begin{align}
	\E\big[\CRPS_d(G, T)\big] &= \E\Big[e^{-rT}\int_{0}^{\tau} \big(G(t) - \1_{\{t\geq T\}}\big)^2\D t\Big] \notag\\
	&=\P\{T<\tau\}\E\Big[e^{-rT}\int_{0}^{\tau} \big(G(t) - \1_{\{t\geq T\}}\big)^2\D t\mid T<\tau\Big] \notag\\
	&\hspace{4cm} +\P\{T=\tau\}\E\Big[e^{-rT}\int_{0}^{\tau} \big(G(t) - \1_{\{t\geq T\}}\big)^2\D t\mid T=\tau\Big] \notag\\
	&=\int_{0}^{\tau}\E\Big[e^{-rT}\big(G(t) - \1_{\{T\leq t\}}\big)^2 \mid T< \tau\Big]\D t\times \underbrace{F(\tau-)}_{\text{event during }[0,\tau)} \notag\\
	&\hspace{4cm} + e^{-r\tau}\int_{0}^{\tau}G^2(t)\D t\times \underbrace{\big[1-F(\tau-)\big]}_{\text{no event during }[0,\tau)}\notag\\
	& = \int_{0}^{\tau}\E\Big[\big(G(t) - \1_{\{T\leq t\}}\big)^2 e^{-rT}\mid T< \tau\Big]F(\tau-)\notag \\
	&\hspace{4cm} + G^2(t) \big[1-F(\tau-)\big]e^{-r\tau}\D t \notag\\
	&= \int_{0}^{\tau}\E\Big[\big(G^2(t) - 2G(t)\1_{\{T\leq t\}} + \1_{\{T\leq t\}}\big) e^{-rT}\mid T< \tau\Big]F(\tau-) \notag\\
	&\hspace{4cm}+ G^2(t) \big[1-F(\tau-)\big]e^{-r\tau}\D t \notag\\
	&= \int_{0}^{\tau}G^2(t)\E\big[e^{-rT}\mid T< \tau\big]F(\tau-) - 2G(t)\E\big[\1_{\{T\leq t\}}e^{-rT}\mid T< \tau\big]F(\tau-) \notag\\
	&\hspace{2cm}+ \E\big[\1_{\{T\leq t\}}e^{-rT}\mid T< \tau\big]F(\tau-) + G^2(t) \big[1-F(\tau-)\big]e^{-r\tau}\D t.\label{eq:(tbm)}
\end{align}
Invoke standard results from the calculus of variations and define the above integrand to be the Lagrange function, i.e., 
\begin{multline}\label{eq:Lagrange}
	\mathcal{L}\big(t, G(t)\big) := G^2(t)\Big\{\E\big[e^{-rT}\mid T< \tau\big]F(\tau-) + \big[1-F(\tau-)\big]e^{-r\tau}\Big\} \\
	- 2G(t)\E\big[\1_{\{T\leq t\}}e^{-rT}\mid T< \tau\big]F(\tau-) + \E\big[\1_{\{T\leq t\}}e^{-rT}\mid T< \tau\big]F(\tau-).
\end{multline}
Specifically, we apply the Euler--Lagrange calculus as set out by \citet{Dac08}; see, in particular, the final part of his Theorem~4.12 and Example~4.15~(ii).\footnote{Observe for the application of \citeauthor{Dac08}'s \citeyearpar{Dac08} Theorem~4.12 that his $f(x,u,\xi)$ corresponds to our Lagrange function $\mathcal{L}(t,G(t))$, with $t$ playing the role of $x$, and $G(t)$ that of $\xi$.}
Note that $G\mapsto\mathcal{L}(t,G)$ is convex for every $t\in(0,\tau)$.
Also, condition (H3) of \citet{Dac08} is satisfied because, for some sufficiently large $K>0$ that may change from line to line,
\begin{align*}
	\big|\mathcal{L}(t,G)\big| & \leq G^2 K + 2G K + K \leq K + K G^2,\\
	\bigg|\frac{\partial}{\partial G}\mathcal{L}(t,G)\bigg| & \leq 2 G K + 2 K \leq K + K G,
\end{align*}
where we used that $\E\big[e^{-rT}\mid T< \tau\big]$ is bounded.
Hence, the aforementioned results in \citet{Dac08} (with, in his notation, $p=2$) imply that the minimizer $G(\cdot)$ of \eqref{eq:(tbm)} satisfies the Euler--Lagrange equation
\begin{align*}
	\frac{\partial \mathcal{L}(t,G)}{\partial G} &= 2G(t)\Big\{\E\big[e^{-rT}\mid T< \tau\big]F(\tau-) + \big[1-F(\tau-)\big]e^{-r\tau}\Big\} - 2\E\big[\1_{\{T\leq t\}}e^{-rT}\mid T< \tau\big]F(\tau-) \\
	& =c
\end{align*}
for some $c\in\mathbb{R}$.
Therefore, 
\begin{equation}\label{eq:(14-)}
	G(t) =\frac{c/2+\E\big[\1_{\{T\leq t\}}e^{-rT}\mid T< \tau\big]F(\tau-)}{\E\big[e^{-rT}\mid T< \tau\big]F(\tau-) + \big[1-F(\tau-)\big]e^{-r\tau}}.
\end{equation}
Since $T\in[0,\tau]$ it has to hold by \eqref{eq:(14-)} that
\begin{align*}
	1&= G(\tau) \\
	&= \frac{\E\big[\1_{\{T\leq\tau\}}e^{-rT}\mid T< \tau\big]+c/2}{F(\tau-)\E\big[e^{-rT}\mid T< \tau\big] + \big[1-F(\tau-)\big]e^{-r\tau}}\\
	&= 1 + \frac{c/2}{F(\tau-)\E\big[e^{-rT}\mid T< \tau\big] + \big[1-F(\tau-)\big]e^{-r\tau}},
\end{align*}
it necessarily has to be the case that $c/2=\big[1-F(\tau-)\big]e^{-r\tau}$.
Therefore, putting $c/2=\big[1-F(\tau-)\big]e^{-r\tau}$ in \eqref{eq:(14-)} gives by the law of total expectation that
\begin{align}
	G(t) &= \frac{\big[1-F(\tau-)\big]e^{-r\tau} + \E\big[\1_{\{T\leq t\}}e^{-rT}\mid T< \tau\big]F(\tau-)}{\big[1-F(\tau-)\big]e^{-r\tau} + \E\big[e^{-rT}\mid T< \tau\big]F(\tau-)}\notag\\
	&= \frac{\E\big[\1_{\{T\leq t\}}e^{-rT}\big]}{\E\big[e^{-rT}\big]}\notag\\
	&=F(t)\frac{\E\big[\1_{\{T\leq\tau\}}e^{-rT}\big]}{\E\big[e^{-rT}\big]}.\label{eq:(def G)}
\end{align}
It is obvious that $G$ is a proper cdf.

We now demonstrate that $G$ is the unique minimizer of \eqref{eq:(tbm)}.
To do so, we invoke Theorem~4.1 of \citet{Dac08}, where his $f(x,u,\xi)$ again corresponds to our Lagrange function $\mathcal{L}(t,G)$.
Write
\[
	\mathcal{L}(t,G) = aG^2 -2bG + c 
\]
for implicitly defined positive $a,b,c$; cf.~\eqref{eq:Lagrange}.
Condition (H1) of Theorem~4.1 is satisfied by strict convexity of $G\mapsto\mathcal{L}(t,G)$, since
\[
	\frac{\partial^2\mathcal{L}(t,G)}{(\partial G)^2}= 2a>0.
\]
To check (H2) of \citet[Theorem~4.1]{Dac08}, it follows from straightforward computations that one can find $\alpha>0$ and $\gamma\in\mathbb{R}$ (e.g., $\alpha=a/2$ and $\gamma=c-2b^2/a$) such that
\[
	\mathcal{L}(t,G) = aG^2 -2bG + c \geq \alpha G^2 + \gamma.
\]
Then, (H2) is immediate.
By Theorem~4.1 in combination with Remark~4.2~(i) of \citet{Dac08}, uniqueness of the solution now follows as $G\mapsto \mathcal{L}(t,G)$ is strictly convex.

Finally, we prove $G(t)\geq F(t)$.
By \eqref{eq:(def G)}, $G(t)\geq F(t)$ is evidently satisfied because $t\mapsto\E\big[e^{-rT}\mid T\leq t\big]$ is non-decreasing, and $\E\big[e^{-rT}\mid T\leq \tau\big]=\E\big[e^{-rT}\big]$ (since $T\in[0,\tau]$).
\end{proof}

\subsection{Proof of Equation~\eqref{M-eq:G r inf}}\label{sec:proof eqn}

For the proof of equation~\eqref{M-eq:G r inf}, the following lemma will be useful.

\begin{lem}\label{lem:exp ratio limit}
Let $Y\geq0$ satisfy that $\essinf Y=0$.
Then, for every $\varepsilon>0$,
\[
	\lim_{r\to\infty}\frac{\E\big[e^{-rY}\1_{\{Y>\varepsilon\}}\big]}{\E\big[e^{-rY}\big]}=0.
\]
\end{lem}

\begin{proof}
Let $\varepsilon>0$.
Since $e^{-rY}\leq e^{-r\varepsilon}$ on the set $\{Y>\varepsilon\}$, it holds that
\begin{equation}\label{eq:E1}
	\E\big[e^{-rY}\1_{\{Y>\varepsilon\}}\big]\leq e^{-r\varepsilon}\P\{Y>\varepsilon\}.
\end{equation}
Moreover, on the set $\{Y\leq \varepsilon/2\}$ we have that $e^{-rY}\geq e^{-r\varepsilon/2}$, such that
\begin{align}
	\E\big[e^{-rY}\big] & = \E\big[e^{-rY}\1_{\{Y\leq\varepsilon/2\}}\big] + \E\big[e^{-rY}\1_{\{Y>\varepsilon/2\}}\big]\notag\\
	&\geq \E\big[e^{-rY}\1_{\{Y\leq\varepsilon/2\}}\big]\notag\\	
	&=e^{-r\varepsilon/2}\P\{Y\leq \varepsilon/2\}.\label{eq:E2}
\end{align}
Note here that $\P\{Y\leq \varepsilon/2\}>0$ because $\essinf Y=0$. 
Combining \eqref{eq:E1}--\eqref{eq:E2}, we obtain that
\[
	0\leq \frac{\E\big[e^{-rY}\1_{\{Y>\varepsilon\}}\big]}{\E\big[e^{-rY}\big]}\leq\frac{\P\{Y> \varepsilon\}}{\P\{Y\leq \varepsilon/2\}}e^{-r\varepsilon/2}\longrightarrow0,
\]
as $r\to\infty$.
\end{proof}

\begin{proof}[\textbf{Proof of Equation~\eqref{M-eq:G r inf}:}]
Recall the definition $a=\essinf T$.
We study the limit of 
\[
	G(t)=\frac{\E\big[e^{-rT}\1_{\{T\leq t\}}\big]}{\E\big[e^{-rT}\big]}.
\]
To do so, we distinguish the three cases $t<a$, $t>a$ and $t=a$.

\textbf{1st case:} $t<a$

In this case, $\1_{\{T\leq t\}}=0$ almost surely, such that $G(t)=0$ and, hence, $\lim_{r\to\infty}G(t)=0$ follows trivially.

\textbf{2nd case:} $t>a$

Put $Y:=T-a$, whence $Y$ satisfies the assumptions of Lemma~\ref{lem:exp ratio limit}.
Choose $\varepsilon>0$ such that $a+\varepsilon<t$.
Therefore, $\{Y\leq\varepsilon\}\subset\{T\leq t\}$.
Then,
\begin{align*}
	\E\big[e^{-rY}\1_{\{T\leq t\}}\big]&=\E\big[e^{-rY}\big] - \E\big[e^{-rY}\1_{\{T> t\}}\big]\\
	&\geq \E\big[e^{-rY}\big] - \E\big[e^{-rY}\1_{\{Y> \varepsilon\}}\big].
\end{align*}
Dividing by $\E\big[e^{-rY}\big]$ gives that
\[
	G(t)=\frac{\E\big[e^{-rY}\1_{\{T\leq t\}}\big]}{\E\big[e^{-rY}\big]}\geq 1-\frac{\E\big[e^{-rY}\1_{\{Y> \varepsilon\}}\big]}{\E\big[e^{-rY}\big]}.
\]
By Lemma~\ref{lem:exp ratio limit}, the second term on the right-hand side tends to 0, as $r\to\infty$.
Since, moreover, $0\leq G(t)\leq 1$, we conclude that $\lim_{r\to\infty}G(t)=1$.

\textbf{3rd case:} $t=a$

Here, we further distinguish two subcases.

\textbf{(i)} $\P\{T=a\}=0$

Then, for every $\varepsilon>0$,
\[
	\{T\leq a\}\subset\{Y\leq 0\}\subset\{Y\leq\varepsilon\},
\]
and Lemma~\ref{lem:exp ratio limit} implies
\[
	0\leq\frac{\E\big[e^{-rT}\1_{\{T\leq a\}}\big]}{\E\big[e^{-rT}\big]}\leq \frac{\E\big[e^{-rY}\1_{\{Y\leq \varepsilon\}}\big]}{\E\big[e^{-rY}\big]}\longrightarrow0,
\]
as $r\to\infty$.
Therefore, $G(t)=G(a)=0$.

\textbf{(ii)} $\P\{T=a\}>0$

Then,
\[
	\E\big[e^{-rT}\1_{\{T\leq a\}}\big] = e^{-ra}\P\{T=a\},
\]
such that
\begin{equation}\label{eq:one}
	\E\big[e^{-rT}\big] = \E\big[e^{-rT}\1_{\{T\leq a\}}\big] + \E\big[e^{-rT}\1_{\{T> a\}}\big]\geq e^{-ra}\P\{T=a\}.
\end{equation}
Moreover, as $r\to\infty$,
\begin{equation}\label{eq:two}
	\E\big[e^{-rT}\1_{\{T> a\}}\big] = o(e^{-ra}).
\end{equation}
To see this, note that, because $\E\big[e^{-rT}\1_{\{T> a\}}\big]=e^{-ra}\E\big[e^{-rY}\1_{\{Y> 0\}}\big]$, this is equivalent to
\[
	\E\big[e^{-rY}\1_{\{Y> 0\}}\big]=o(1).
\]
However, since $e^{-rY}\1_{\{Y> 0\}}\to0$, as $r\to\infty$, and $e^{-rY}\1_{\{Y> 0\}}\leq 1$, this follows immediately by dominated convergence.

Now, combining \eqref{eq:one}--\eqref{eq:two}, we obtain that
\begin{align*}
	G(t)=G(a)& =\frac{\E\big[e^{-rT}\1_{\{T\leq a\}}\big]}{\E\big[e^{-rT}\big]}\\
	&= \frac{\E\big[e^{-rT}\big] - \E\big[e^{-rT}\1_{\{T> a\}}\big]}{\E\big[e^{-rT}\big]} \\
	& = 1-\frac{\E\big[e^{-rT}\1_{\{T>a\}}\big]}{\E\big[e^{-rT}\big]}\\
	& \longrightarrow 1,
\end{align*}
as $r\to\infty$.
\end{proof}

\section{Methodology}\label{sec:Methodology}

\textbf{Intention:} Our primary goal in the empirical section is not to measure an effect size with precision but to investigate whether time preference effects, i.e., discount-driven misreports, are a problem in practice. To that end, we treat the naturally occurring mix of time-fixed and time-varying questions on Metaculus as a setting that is similar to a natural experiment and try to inch as close as possible to a quasi-experimental setup. This approach trades off some statistical ``cleanliness'' in favor of clearer identification of behavioural patterns. Since the effect size depends on numerous features, such as time-to-event and personal discount rates, we do not expect the measured effect size to be a useful approximation of the size of time preference effects in other settings. 

\textbf{Replicability:} To replicate this study, researchers need three resources: (i) the data-analysis code, (ii) the Metaculus forecasting data, and (iii) the relevant supplementary data from the Metaculus API. We address each in turn. First, we provide the data-analysis code in the following repository:
 
\begin{quote}
\url{https://github.com/PfadQualle/time_preferences_in_forecasting}
\end{quote}
 
Second, although the Metaculus API is publicly available, it has changed substantially since we wrote our \texttt{R} code in early 2025. As a result, our code may no longer run as intended with the current version of the API. To address this problem, we have included all relevant data retrieved from the API in the GitHub repository, i.e., accessing the Metaculus API is not needed.
 
Third, the Metaculus forecasting data was generously provided to us, but it is proprietary, and we are therefore unable to release it publicly. Researchers who wish to obtain it should contact Metaculus directly at \url{https://metaculus.com}.

\subsection{Additional detail related to data and sample construction}\label{sec:Additional detail related to data and sample construction}

\textbf{Absence of time-to-event measures:} Our analysis deliberately omits any time-to-event measure. Intuitively, one might wish to control for such a measure, since a forecaster’s discounting should depend on their expected time-to-event or, more specifically, their expected time-to-reward. However, we do not observe any forecaster-specific expectations of time-to-event. In the case of time-varying (unscheduled) questions, the relevant horizon is precisely the expert’s own (unobserved) expectation of when the event will happen. Because this expectation is not recorded in the Metaculus dataset, we cannot include a time-to-event covariate. We note this as a limitation and also caution against simply conditioning on the realized time-to-event as this will introduce a ``surprisingly early bias'' in the data.% \citep{lehmann_surprisingly-early_2025}. 
For example, suppose we were to split up all forecasts based on their time-to-event, say, we group forecasts with a time-to-event of less than 1 year in one set and forecasts with a time-to-event of more than 1 year in another set. 
Then, the first group would inevitably include a number of events that occurred sooner than honestly expected (surprisingly early), while the second group would, conversely, be populated by events whose outcomes remained unknown for longer than honestly expected (surprisingly late). This means that, even for a forecaster whose predictions are perfectly calibrated, calibration will not be perfect within each of these separate groups, but only when considering all forecasts together. By splitting the data in this manner, we inadvertently introduce a selection bias that confounds our effect of interest.

\textbf{Classification:} Since the distinction in time-varying and time-fixed events is new, and not made in the Metaculus data, we classified all binary events by hand using the ``resolution criteria'' provided by the question authors. This led to an unambiguous classification in most cases. We applied the label ``NA'' to ambiguous resolution criteria.

\textbf{Time window:} We used all binary questions kindly provided to us by Metaculus, from the inception of the platform in 2015 until our data cutoff in August 2024. As already noted in the paper, we needed to exclude a total of 35 events in order to eliminate the ``surprisingly early bias''. % We refer to \citet{lehmann_surprisingly-early_2025} for details. 

\textbf{Focus on binary questions:} We focus on binary questions in this analysis for multiple reasons. Firstly, we wanted to keep the analysis as simple as possible. Introducing distributional forecasts would have increased complexity. Secondly, distributional forecasts (on Metaculus) are more often on time-varying events. Thus, these events would not have had the properties that we seek in a quasi-experiment analysis, such as (rough) equality of treatment and control. 

\subsection{Additional details related to the regression analysis}\label{sec: Detail_regression}

\textbf{Linear model:} Some reasons for choosing a linear model are outlined in the main text. Although probabilistic forecasts lie between 0 and 1---suggesting a logit or probit model---we found that generalized linear models with a logistic form produced extreme overfits for many individual forecasters. The linear model, by contrast, fit the data points more ``sensibly''. We concede that this choice allows predictions outside [0,1] in principle, but in practice the fitted values remained in a reasonable range, and the linear model allows for an easier interpretation of the coefficients.

\textbf{Dependent variable:} One point that may cause confusion is our use of the actually observed event frequency $\mu$ as the dependent variable, rather than the reported prediction. This choice reflects the definition of calibration: we model the distribution $\mu\mid G,tvarying$ rather than $G\mid\mu,tvarying$. However, it seems to go ``against the direction of causality'', as reported predictions change as a result of time preferences (variable $tvarying$), not the actual frequency of events. However, since we estimate how often events occur among all forecasts assigned a given probability $G$, the frequency does change when we condition on $tvarying$.

\textbf{Hypothesis tests:} Particular interest attaches to formally testing the joint hypothesis that the calibration for time-fixed event predictions is identical to that for time-varying events, i.e., that both \textit{lines} are not significantly different from each other. In our setup, each coefficient estimate is estimated via a separate random-effects meta-analysis. Because these are four univariate (rather than multivariate) models, standard $F$-tests for the significance of multiple coefficients within models are therefore not applicable. Instead, we apply the Bonferroni correction to test joint hypotheses.

\textbf{Incorporating uncertainty in balanced topics:} Because each weight calculated via inverse probability of treatment weighting (IPTW) depends on a probability to be time-varying, which is itself an estimate, we incorporate this extra layer of uncertainty by bootstrapping the entire analysis, as described by \citet{austin_variance_2016}. Specifically, we draw 221 bootstrap samples, resampling the propensity scores every time, and re-run the entire analysis in each replicate. The empirical distribution of the coefficient estimates shows that the additional variance from propensity-score estimation is near-zero and negligible.

\subsection{Additional details related to the IPTW}\label{Additional details related to the Inverse Probability of Treatment Weighting}

\textbf{Goal of the IPTW:} We implement IPTW to mitigate two sources of confounding. Firstly, time‐fixed events (e.g., elections) and time‐varying events (e.g., eruption occurrences) differ systematically in subject matter, which could have affected results, as discussed in the main text. Secondly, across-event dependence looms as a stealth threat to our attempts to provide a \textit{ceteris paribus} comparison. Across‐event dependence refers to the fact that multiple forecasts may hinge on the same real‐world storyline, as the outcomes are connected in some way. For example, had the Russia-Ukraine war played out differently, or come to an early end, multiple outcomes would have been affected by this. Since across-event dependence is likely the strongest when topics are similar, balancing the forecasts on those topics should affect the influence of across-event dependence. As applying the IPTW weights to our regressions does not meaningfully change the estimated gap between time‐fixed and time‐varying forecasts, we conclude that neither topic imbalance nor across‐event dependence is the primary driver of the observed effect.

\textbf{Implementation:} We implement inverse‐probability‐of‐treatment weighting on the event-level, not on the prediction-level. Topic labels (and legacy tags) were pulled from the Metaculus API.\footnote{\url{https://www.metaculus.com/api/}} Although Metaculus no longer assigns tags, we conjecture that controlling for the remaining topic labels alone produces virtually identical results.

To estimate the propensity score---the probability to be time-varying---we fit a logistic regression
\[
\P(tvarying=1 \mid \mathbf{topics}) = \frac{1}{1 + e^{- \left( \beta_0 + \bm{\beta}^\top \mathbf{topics} \right)}}.
\label{eq:logistic}
\]
We use 80 topics and tags as covariates in the logistic regression to estimate the probability of being a time-varying event. 
Fitted propensities are used to compute weights as detailed in the main text.
We report regression outcomes in Table~\ref{tab:longtable_logistic}. We only use topics and tags that come up at least with five different events. However, some tags come up only with time-fixed or time-varying events. We still keep these tags and topics in the regression. In Table~\ref{tab:longtable_logistic} the coefficients of these topics have very high or very low values---usually around positive or negative 18---and exorbitant standard errors. This phenomenon arises because, theoretically, their estimated size would be infinitely large/small, and the statistical software does not handle this extremity. Fortunately, these anomalies do not undermine our analysis in any way. 
To see this, events associated with topics with very high coefficients always are time-varying, resulting in a propensity score of one, and thus a weight of one. Conversely, for topics associated solely with time-fixed events---reflected in strongly negative coefficients---the propensity score is close to zero, which also produces a weight of one. In all cases, the integrity of our study is preserved.

% in the body of the appendix:
\begin{center}
\scriptsize
\begin{longtable}{@{\extracolsep{\fill}}lrr}
\caption{Coefficient estimates for the logistic regression in \eqref{eq:log reg}}\label{tab:longtable_logistic}\\
\toprule
Covariate & Estimate & Std.\ Error \\
\midrule
\endfirsthead

\multicolumn{3}{l}{\textit{Table \thetable\ continued from previous page}} \\
\toprule
Covariate & Estimate & Std.\ Error \\
\midrule
\endhead

\midrule
\multicolumn{3}{r}{\textit{Continued on next page}} \\
\endfoot

\bottomrule
\endlastfoot

(Intercept)                                   &$  0.96519$   & $0.02715$ \\
Technology                                    &$  0.79915$   & $0.04080$ \\
Social Sciences                               &$  0.44501$   & $0.07431$ \\
Economy \& Business                           &$  0.17403$   & $0.03219$ \\
S\&P 500 Index                                &$ 17.42685$   & $427.31226$ \\
Natural Sciences                              &$  1.14779$   & $0.05605$ \\
California                                    &$  0.66274$   & $0.11860$ \\
Bioethics                                     &$  0.15485$   & $0.10683$ \\
Computing and Math                            &$ -1.47888$   & $0.04805$ \\
Politics                                      &$ -0.93215$   & $0.02860$ \\
Russia                                        &$  1.84676$   & $0.03667$ \\
Geopolitics                                   &$  0.19724$   & $0.03468$ \\
Health \& Pandemics                           &$  0.13689$   & $0.03922$ \\
Epidemiology                                  &$  1.35235$   & $0.04919$ \\
SpaceX                                        &$ -0.90372$   & $0.05607$ \\
Cryptocurrencies                              &$  1.02433$   & $0.07922$ \\
Medicine                                      &$  0.41705$   & $0.06986$ \\
China                                         &$ -0.12499$   & $0.11287$ \\
Ukraine                                       &$  0.08893$   & $0.03573$ \\
US Senate                                     &$  0.17288$   & $0.02780$ \\
Elections                                     &$ -2.56205$   & $0.03470$ \\
Artificial Intelligence                       &$ -1.00066$   & $0.05230$ \\
Joe Biden                                     &$  2.88609$   & $0.07449$ \\
Virology                                      &$  1.32127$   & $0.04928$ \\
Elon Musk                                     &$ -1.76658$   & $0.05229$ \\
Donald Trump                                  &$  0.66523$   & $0.02243$ \\
President of the US                           &$  0.17121$   & $0.06244$ \\
DeepMind                                      &$  0.34760$   & $0.06051$ \\
Google Scholar                                &$ 16.45309$   & $365.19765$ \\
Sports \& Entertainment                       &$ -1.30713$   & $0.03766$ \\
LIGO                                          &$ 16.51459$   & $453.73149$ \\
United States                                 &$  0.81035$   & $0.04910$ \\
North Korea                                   &$  0.12224$   & $0.06217$ \\
Planetary science                             &$ 16.45309$   & $273.92508$ \\
Space                                         &$ -0.52675$   & $0.10393$ \\
Tesla Model 3                                 &$ -1.11258$   & $0.10596$ \\
Bureau of Labor Statistics                    &$ -0.16812$   & $0.10023$ \\
Law                                           &$ -0.09966$   & $0.05657$ \\
Aerospace engineering                         &$ -1.73657$   & $0.08066$ \\
Environment \& Climate                        &$ -1.34214$   & $0.07626$ \\
Microsoft                                     &$ 18.03582$   & $127.95781$ \\
Benjamin Netanyahu                            &$ -0.46429$   & $0.06081$ \\
Boris Johnson                                 &$ -0.93256$   & $0.04289$ \\
Twitter                                       &$ 17.29658$   & $170.50745$ \\
Chatbot                                       &$  0.84639$   & $0.09484$ \\
Silicon Valley Bank                           &$ 17.42685$   & $193.43858$ \\
SEC                                           &$  1.90004$   & $0.12091$ \\
OpenAI                                        &$  1.56564$   & $0.09720$ \\
Hunter Biden                                  &$  0.50166$   & $0.05245$ \\
New York (state)                              &$ -0.91728$   & $0.05936$ \\
PredictIt                                     &$ -1.41526$   & $0.07925$ \\
Tesla                                         &$  1.15337$   & $0.16105$ \\
Washington, D.C.                              &$  1.28104$   & $0.13432$ \\
Republican Party (US)                         &$ -0.87933$   & $0.08391$ \\
Australia women's national soccer team        &$ -1.30069$   & $0.08839$ \\
Climate change denial                         &$ -1.64016$   & $0.11217$ \\
Democratic Left Alliance                      &$-18.59911$   & $310.25028$ \\
FIDE world rankings                           &$-18.22413$   & $125.78473$ \\
Saudi Professional League                     &$-18.22413$   & $468.29816$ \\
Yesh Atid                                     &$-18.59911$   & $214.34720$ \\
Kukiz'15                                      &$-16.96921$   & $340.01604$ \\
Pieter Omtzigt                                &$-18.59911$   & $423.69091$ \\
European United Left--Nordic Green Left       &$-18.59911$   & $490.27154$ \\
Roy Jenkins                                   &$-18.59911$   & $410.07478$ \\
S.C. Braga                                    &$-18.22413$   & $197.11345$ \\
Phil Murphy                                   &$-19.53126$   & $273.92507$ \\
Donald Trump Jr.                              &$-16.96921$   & $62.35268$ \\
Ilya Sutskever                                &$-20.33041$   & $239.61512$ \\
Ecuador                                       &$-18.44799$   & $271.91233$ \\
Richárd Rapport                               &$-18.22413$   & $230.89891$ \\
Roman Baber                                   &$-18.59911$   & $387.73062$ \\
Param Singh                                   &$-16.96921$   & $613.59823$ \\
Meral Akşener                                 &$-18.59911$   & $104.23222$ \\
Liz Truss                                     &$-18.59911$   & $119.14613$ \\
Nate Silver                                   &$-17.26262$   & $253.32567$ \\
African National Congress                     &$-18.59911$   & $1304.52772$ \\
Jake Sherman                                  &$-18.59911$   & $1331.42803$ \\
Brian Fitzpatrick                             &$-18.59911$   & $112.52621$ \\
Patrick McHenry                               &$-18.59911$   & $192.25855$ \\
List of MLB progressive HR leaders            &$-18.22413$   & $418.42746$ \\
\end{longtable}
\end{center}

Since the probability to be time-varying is calculated from the event dataset, this means that we do not balance topics \textit{within} individual forecasters, but across forecasters. However, we do check for balance \textit{within} forecasters to make sure that IPTW works to balance differences between treatment and control. Some forecasters may still exhibit unbalanced panels of predictions, but the average forecasters, and average effect (that we are interested in) will be unbiased.

\begin{figure}[htbp]
  \centering
  \includegraphics[width=\textwidth]{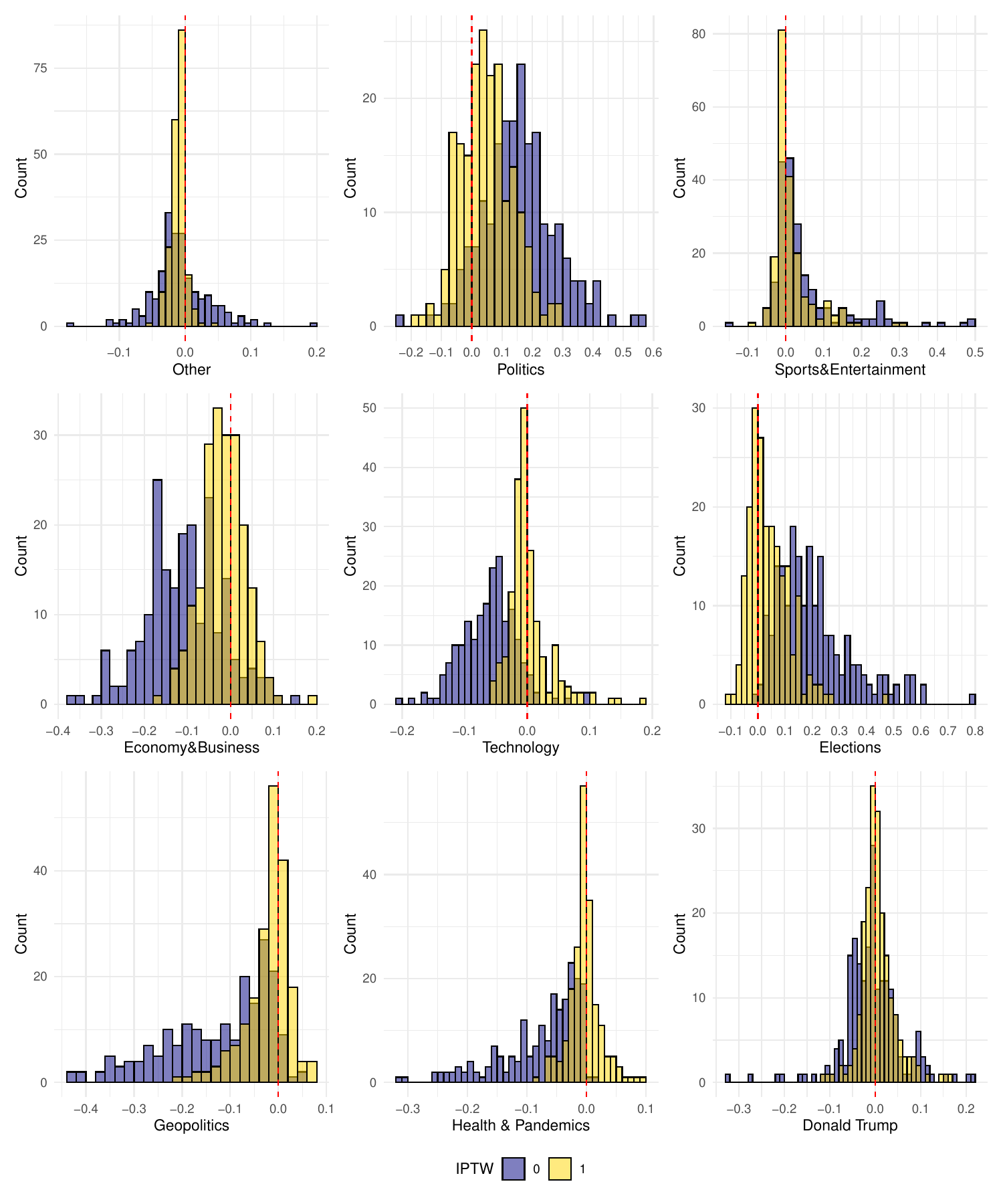}
  \caption{Differences in topic share within forecasters’ predictions on time-varying vs.\ time-fixed events.}
  \label{fig:Balance_1}
\end{figure}

\begin{figure}[htbp]
  \centering
  \includegraphics[width=\textwidth]{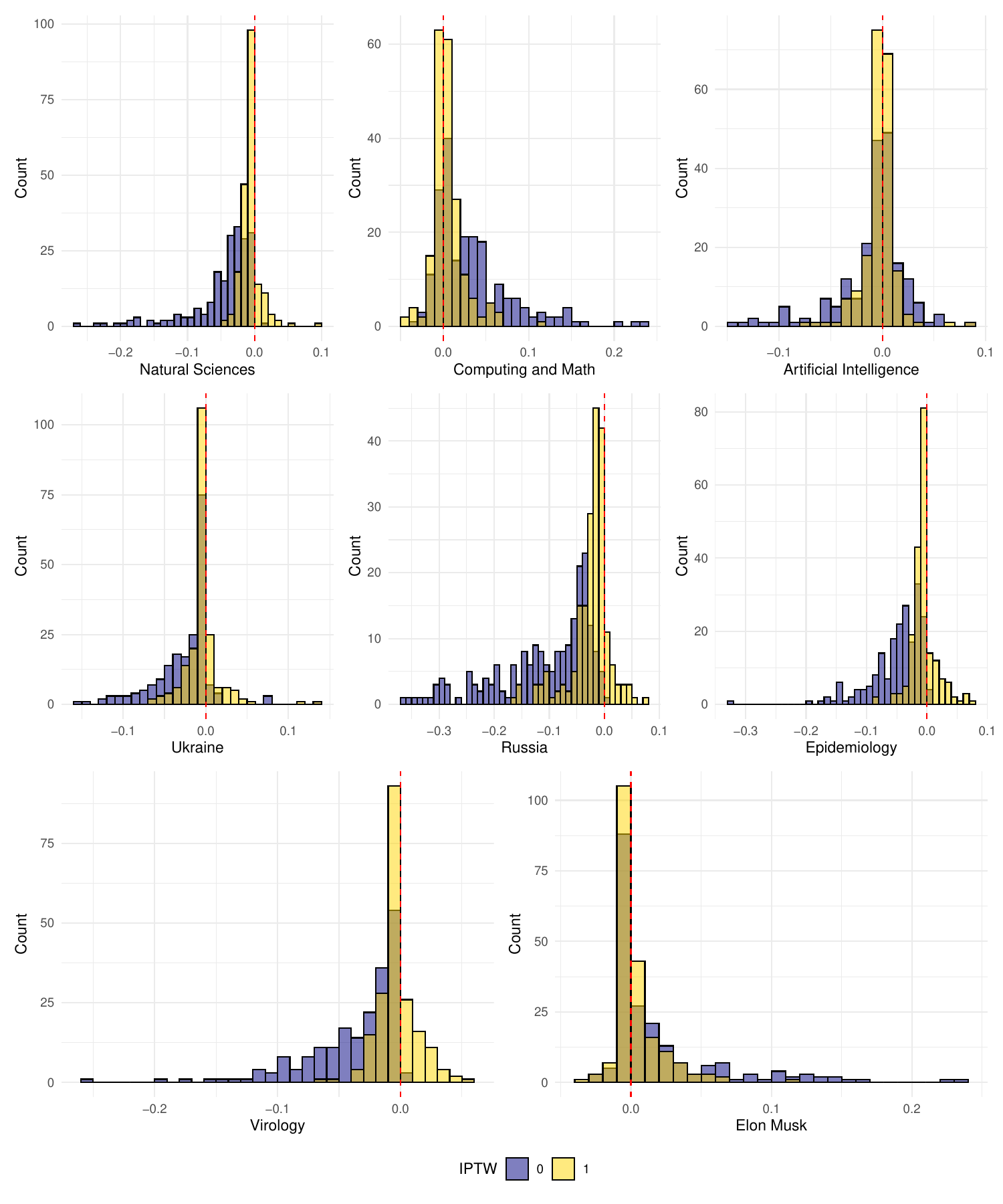}
  \caption{Differences in topic share within forecasters’ predictions on time-varying vs.\ time-fixed events.}
  \label{fig:Balance_2}
\end{figure}

%%%% Differences in predictions within forecasters
Since we are interested whether predictions on time-varying and time-fixed events are systematically different (and we introduce a within-subject-design in Section~\ref{M-sec: Regression analysis}) we additionally examine the distribution of topics differences in predictions \textit{within} individual subjects. For each individual forecaster in our sample we calculate the share of time-varying and time-fixed event predictions for each topic. We then calculate the difference between the share of time-varying and time-fixed event predictions on each topic. The distributions of differences across forecasters are presented in Figures~\ref{fig:Balance_1} and \ref{fig:Balance_2}. The navy-blue histogram refers to the unweighted distribution, i.e., the ``original'' distribution of the data. The $x$-axis refers to the difference, i.e., share of time-fixed event predictions minus share of time-varying event predictions:
\begin{equation}\label{eq:log reg}
    \frac{n_{tfixed} - n_{tvarying}}{n_{tfixed} + n_{tvarying}}.
\end{equation} 
Thus, if a topic is represented more often in time-fixed event predictions within forecasters, the histogram will be skewed right. The $y$-axis refers to the count of forecasters within the specific bin-width of the histogram. The graphs in Figures~\ref{fig:Balance_1} and \ref{fig:Balance_2} show that time-fixed and time-varying event predictions are not balanced. Rather, as the aggregate picture in Table~\ref{M-tab:predictions|topics} suggests, forecasters make more predictions on time-fixed events when they forecast topics such as Politics, Elections, Sports \& Entertainment, Computing \& Math, whereas predictions on Geopolitics, Economy \& Business, Health \& Pandemics, Technology and Natural Sciences tend to be predominantly time-varying events.

The golden histogram, which overlays the navy blue one in Figures~\ref{fig:Balance_1} and \ref{fig:Balance_2}, shows the distribution after applying IPTW. Since the golden histograms are much closer to---and more tightly centered around---zero  on the horizontal axis (where forecasters have made as many time-varying predictions as time-fixed predictions), we see clearly that IPTW works to equilibrate differences between predictions on the same topic within forecasters. Individual forecasters may still show imbalance, but the overall average treatment effect should be much less biased.  Applying these weights to our calibration regressions does not materially change the estimated gap between time‐fixed and time‐varying forecasts (Table~\ref{M-tab:regrssn_outcomes}). We therefore conclude that neither topic imbalance nor across‐event dependence is the primary driver of the observed miscalibration.

\section{Robustness checks}\label{sec:Robustness checks}

\subsection{Additional sample}

%%%Context:
We classified 197 events with 28,667 corresponding forecasts as time-varying events that can ``only turn out to not happen early on'', i.e., impose the opposite incentive as events which can only happen early. These include events such as ``Will Germany win the FIFA World Cup in 2018?'', which turned out not to happen before it could have turned out to happen (at the final). We plot the calibration on these events in green in Figure~\ref{fig:calibration_dynamic_-}. We see in Figure~\ref{fig:calibration_dynamic_-} that the forecasts on these events are on average lower than the forecasts on time-fixed events (observed frequency of events is higher conditional on the same prediction), but it is not as clear cut as with time-varying events that can only happen early. This coincides with our theoretical contemplations as the forecasters have an incentive to lower reported predictions to achieve a higher accuracy in case the event does not occur (which happens earlier).

\begin{figure}
    \centering
    \includegraphics[width=0.8\textwidth]{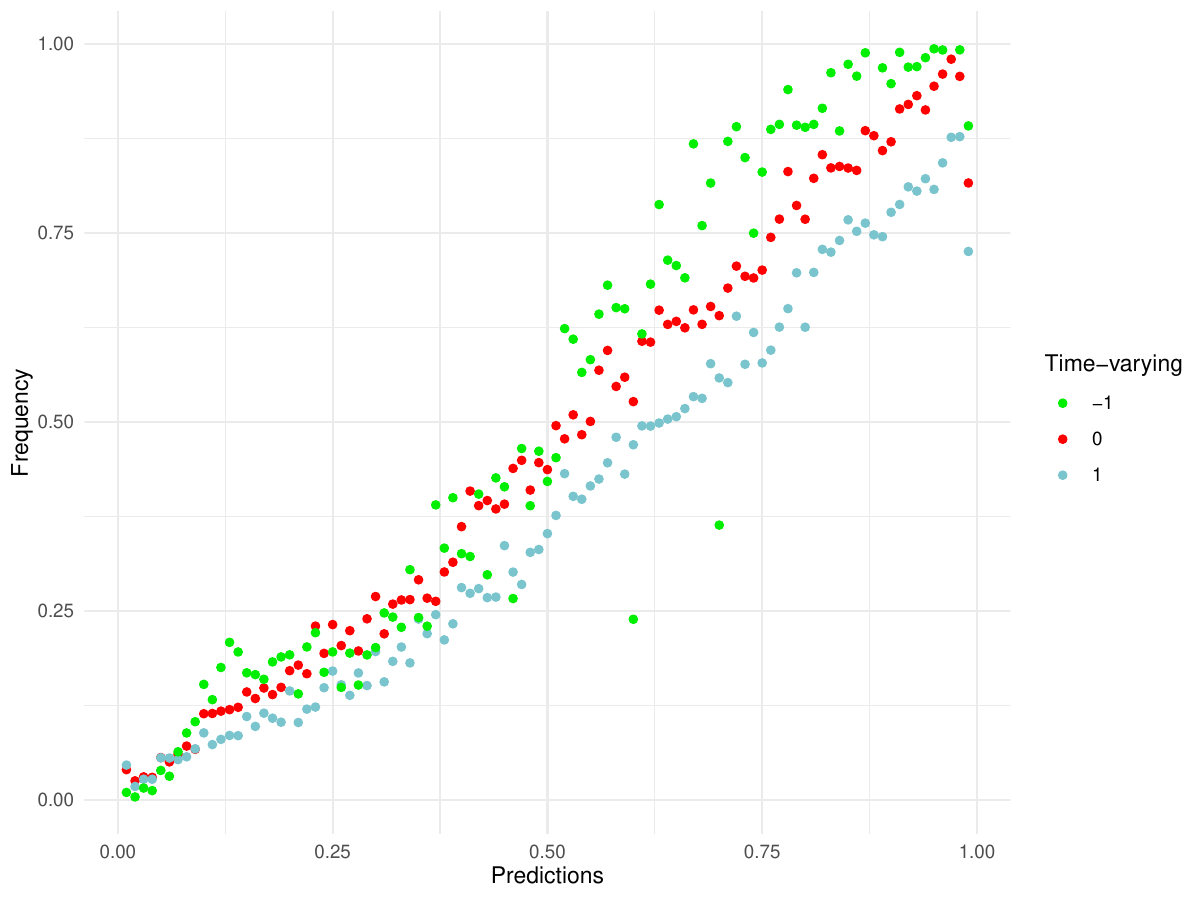}
    \caption{Calibration including time-varying events with opposing incentives}
    \label{fig:calibration_dynamic_-}
\end{figure}

Most events focus on competitions (Super Bowl, Formula 1, etc.). Given the small sample size and topical one-sidedness of these events, we cannot apply the full analysis to this sample. However, we can apply our standard OLS regression on the full sample when we codify the events that can only turn out not to happen early as $tvarying=-1$. The regression equation remains unchanged. The results are reported in Table~\ref{tab:appendix_regrssn_outcomes}, column 2. We find that the inclusion of this sample does not affect the results in any meaningful way.

\subsection{Within-event dependence}\label{Within-event dependence}

Metaculus allows each forecaster to update a probability forecast multiple times for the same event. Although we treat each row as an independent observation, this within-event dependence could---in principle---distort coefficient estimates. 
A potential concern is that multiple forecasts linked to the same outcome are counted separately, such that their outcome is not truly independent (as assumed).

To verify that our main findings do not rest on this data structure, we re-ran the core analysis on two alternative, ``collapsed'' versions of the data. In every case, the regression analysis was the same; only the input dataset differed. 

\begin{enumerate}
    \item First-Forecast only
    \begin{itemize}
        \item For each forecaster--event pair, keep only the very first submitted probability.
        \item Reduces $N$ from 215 581 to 56 081 forecasts.
        \item This scenario increases the average ``time to event,'' so \textit{a priori} we expect any time-preference effect to be larger in magnitude.
    \end{itemize}
        \item Last-Forecast only
    \begin{itemize}
        \item For each forecaster--event pair, keep only the last forecast (can also be the first and last and, hence, the only forecast).
        \item Reduces $N$ from 215 581 to 56 081 forecasts.
        \item This scenario reduces the ``time to event,'' so \textit{a priori} we expect any time-preference effect to be lower in magnitude.
    \end{itemize}
\end{enumerate}

We are able to select the first or last forecast because we have their timestamps. Table \ref{tab:appendix_regrssn_outcomes} reports estimates of the regression coefficients in 
\begin{equation}\label{eq:regression}
 \mu_{i} = \underbrace{\beta_0+\beta_1 \ G}_{\text{time-fixed events}}  + \underbrace{\beta_2 \ tvarying + \beta_3 \ tvarying \times G}_{\text{time-varying events}}  + \varepsilon_{i},\qquad\E[\varepsilon_{i}\mid G,\ tvarying]=0,
\end{equation}
where $\beta_0$ denotes the intercept, and $\beta_1$ the slope for the time-fixed event calibration, and $\beta_2$ and $\beta_3$ measure the change in intercept and slope for time-varying events.

Most importantly, the results are not meaningfully different from the results reported in the main text. Standard errors are of the same order of magnitude as in the full dataset. Quantitative changes are relatively minor, confirming that within-event dependence is not a primary driver of the observed effect. The only major change in significance is with $\beta_2$, which is significantly larger than 0 when the last forecasts are used only. We tentatively assume that the cause of this is a stronger center-bias of time-fixed predictions in this dataset than in the others. This is evidenced by a lower $\beta_0$ and higher $\beta_1$. If the center-bias for time-varying predictions is lower than for time-fixed predictions, then this would explain the change in $\beta_2$. We do not think that this is a cause for concern, as the overall size of the coefficient is still small and time-varying forecasts are systematically too high in this model.

These robustness checks demonstrate that treating each update as a separate observation does not bias our core findings. Whether we use the first forecast or the last forecast, the qualitative results remain unchanged. The principal time preference effect we document persists, and precision of measurement remains high even after collapsing the data.

\begin{table}[t!]
    \centering
    \begin{threeparttable}
    \caption{Regression Estimates}\label{tab:appendix_regrssn_outcomes}
    \begin{tabular}{lcccc}
        \toprule
        & \textbf{(1)} &        \textbf{(2)}                    & \textbf{(3)} & \textbf{(4)} \\
        & \textbf{Original} & \textbf{Additional sample} & \textbf{Collapsed (first)} & \textbf{Collapsed (last)}  \\
        \midrule
         $\beta_0$     & $-0.0507$*** & $-0.047$*** & $-0.05$*** & $-0.0579$*** \\
                       & (0.0055) & (0.007) & (0.0049)  & (0.0049) 
                       \vspace{10pt}\\
            $\beta_1$  & $1.0665$*** & $1.009$ & $1.0824$*** & $1.0987$*** \\
                       & (0.0108) & (0.0106) & (0.011) & (0.0106)
                       \vspace{10pt}\\
            $\beta_2$  & $0.0095$ & $-0.004$ & $0.0084$  & $0.0243$*** \\
                       & (0.0049) & (0.0048) & (0.007) & (0.0064) 
                       \vspace{10pt}\\
            $\beta_3$  & $-0.129$*** & $-0.113$*** & $-0.1238$***& $-0.1145$*** \\
                       & (0.0124) & (0.015) & (0.0128)  & (0.011) \\
        \bottomrule
    \end{tabular}
    \begin{tablenotes}
      \small
      \item \textit{Note.} * $p \leq 0.05$; ** $p < 0.01$; *** $p < 0.001$. 
      \item The first column is simply a copy of the results of the main analysis (IPTW \& Within-subject) from the main text. The second column describes the results from the OLS analysis, including the additional sample. The third and fourth column present the results of the robustness check regarding the within-event dependence, where we collapsed updated predictions. The hypothesis test for $\beta_1$ is against 1; for all other coefficients against 0.  
    \end{tablenotes}
    \end{threeparttable}
\end{table}

\section{Calculating the effect of time preferences on accuracy}\label{sec:Effect_on_accuracy}

In order to calculate the effect of time preferences on accuracy, we use the decomposition of strictly proper scoring rules into calibration and precision / refinement \citep{parmigiani_decision_2009}.\footnote{Other names may be used to describe the same terms. For more details on the decomposition of strictly proper scoring rules, we refer to \citet{brocker_reliability_2009}.} That is, the forecasting error that systematic bias adds to the total error is independent of size of the total error. Therefore, we can simply look at the calibration of time-fixed events versus time-varying events and calculate the error contribution of each, without ever calculating full forecasting errors. Specifically, the decomposed squared error is
\begin{align*}
\QSR
  &= \underbrace{\sum_{G \in \Gamma} \frac{n_G}{N}\,(G - \mu_G)^2}_{\text{Calibration}}
   \;+\;
     \underbrace{\sum_{G \in \Gamma} \frac{n_G}{N}\,\big[G(1-G)\big],}_{\text{Precision (Refinement)}}
\end{align*}
where $\Gamma$ denotes the binned prediction levels $\{0.01,0.02,...,0.99\}$, $G$ refers to the predictions, and $\mu_G$ equals the observed frequency of events $\P(X=1 \mid G)$. For example, we could look at all predictions $G=0.25$ and observe an actual frequency of $\mu_G=\P(X=1\mid G=0.25)= 0.17$ in our data. Since $G=0.25$ refers to a small share of predictions $G$, we denote this share as $\frac{n_G}{N}$, where $N$ is the total number of predictions.

We can calculate the calibration for predictions on time-varying and time-fixed events respectively. We utilize all predictions, that is, the comparison is unweighted (no IPTW) and across all forecasters. We basically take the sum of squared residuals between the data points and the 45 degree diagonal in Figure~\ref{M-fig:calibration}. We find that the difference in squared error between time-varying and time-fixed predictions is $\Delta S = 0.008685$. Without knowing the ``total error'' in the dataset (we have not calculated precision), we can safely assume that forecasters are better than random, which would correspond to an error of $S=0.25$. Thus, the error reduction is \textit{at least} $\frac{0.008685}{0.25}\approx3.5\%$.

\section{Forestplots for the Random-Effects Meta Analysis}\label{sec:Forestplots for the Random-Effects Meta Analysis}

\begin{figure}[b!]
    \centering
    \includegraphics[width=.80\textwidth]{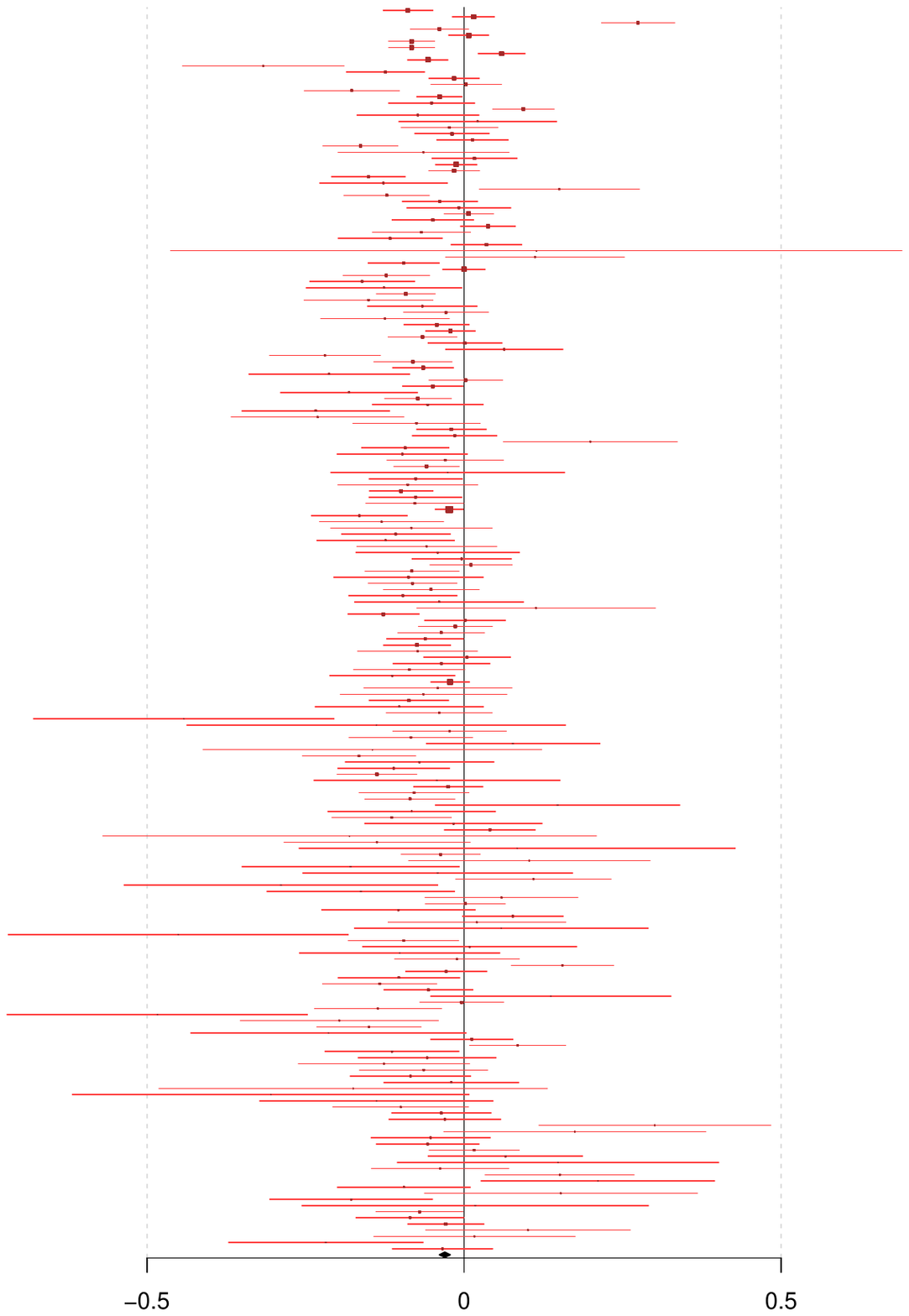}
    \caption{Estimates of the intercept $\beta_0$}
    \vspace{1ex}
            \parbox{0.95\textwidth}{%
      \small
      \textit{Note.} This figure presents the forestplot of individual within-forecaster estimates of the coefficient $\beta_0$, i.e. the intercept. The forecasters are sorted (top to bottom) from most predictions made (total) to least predictions made. The horizontal axis refers to the coefficient's size. Whiskers indicate the 95\%-confidence intervals.
    }\label{FigForest1}
\end{figure}

\begin{figure}
    \centering
    \includegraphics[width=0.85\textwidth]{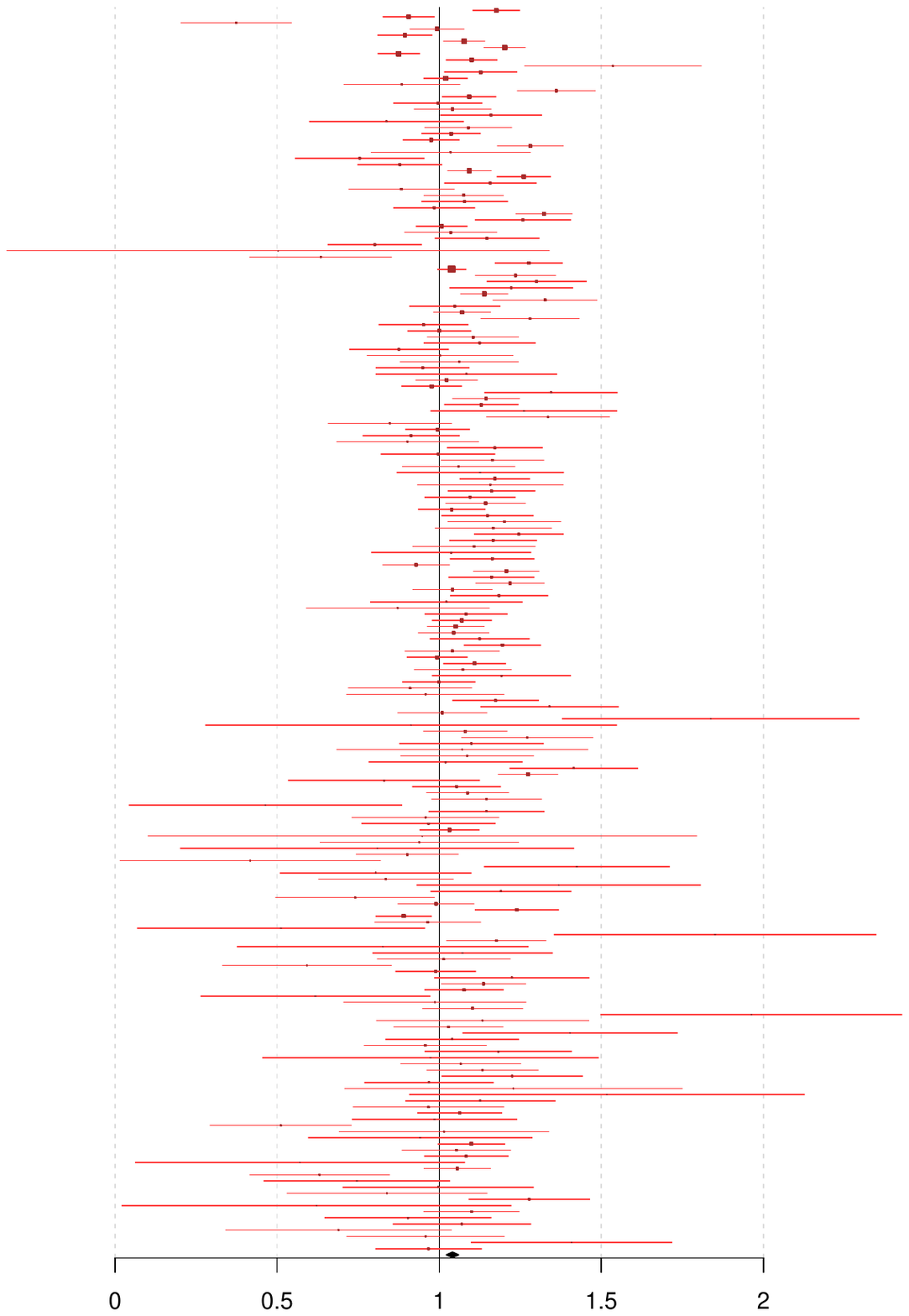}
    \caption{Estimates of the slope $\beta_1$}
    \vspace{1ex}
    \parbox{0.95\textwidth}{%
      \small
      \textit{Note.} This figure presents the forestplot of individual within-forecaster estimates of the coefficient $\beta_1$, i.e. slope, or ``how predictions correspond to frequencies'' for time-fixed events. The forecasters are sorted (top to bottom) from most predictions made (total) to least predictions made. The horizontal axis refers to the coefficient's size. Whiskers indicate the 95\%-confidence intervals.
    }\label{FigForest2}
\end{figure}

\begin{figure}
    \centering
    \includegraphics[width=0.85\textwidth]{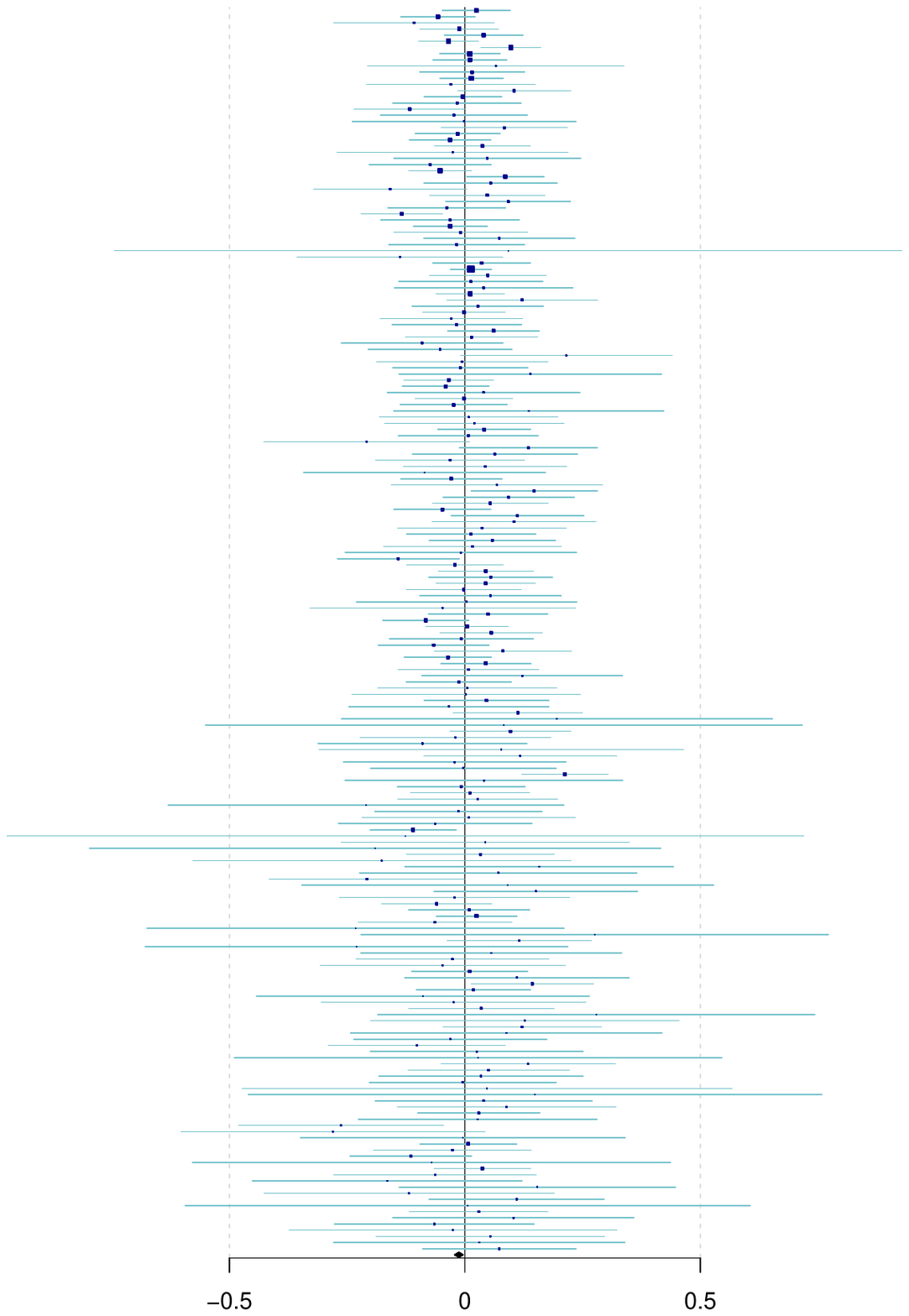}
    \caption{Estimates of the intercept difference $\beta_2$}
        \vspace{1ex}
    \parbox{0.95\textwidth}{%
      \small
      \textit{Note.} This figure presents the forestplot of individual within-forecaster estimates of the coefficient $\beta_2$. The forecasters are sorted (top to bottom) from most predictions made (total) to least predictions made. The horizontal axis refers to the coefficient's size. Whiskers indicate the 95\%-confidence intervals.
    }\label{FigForest3}
\end{figure}

\begin{figure}
    \centering
    \includegraphics[width=0.85\textwidth]{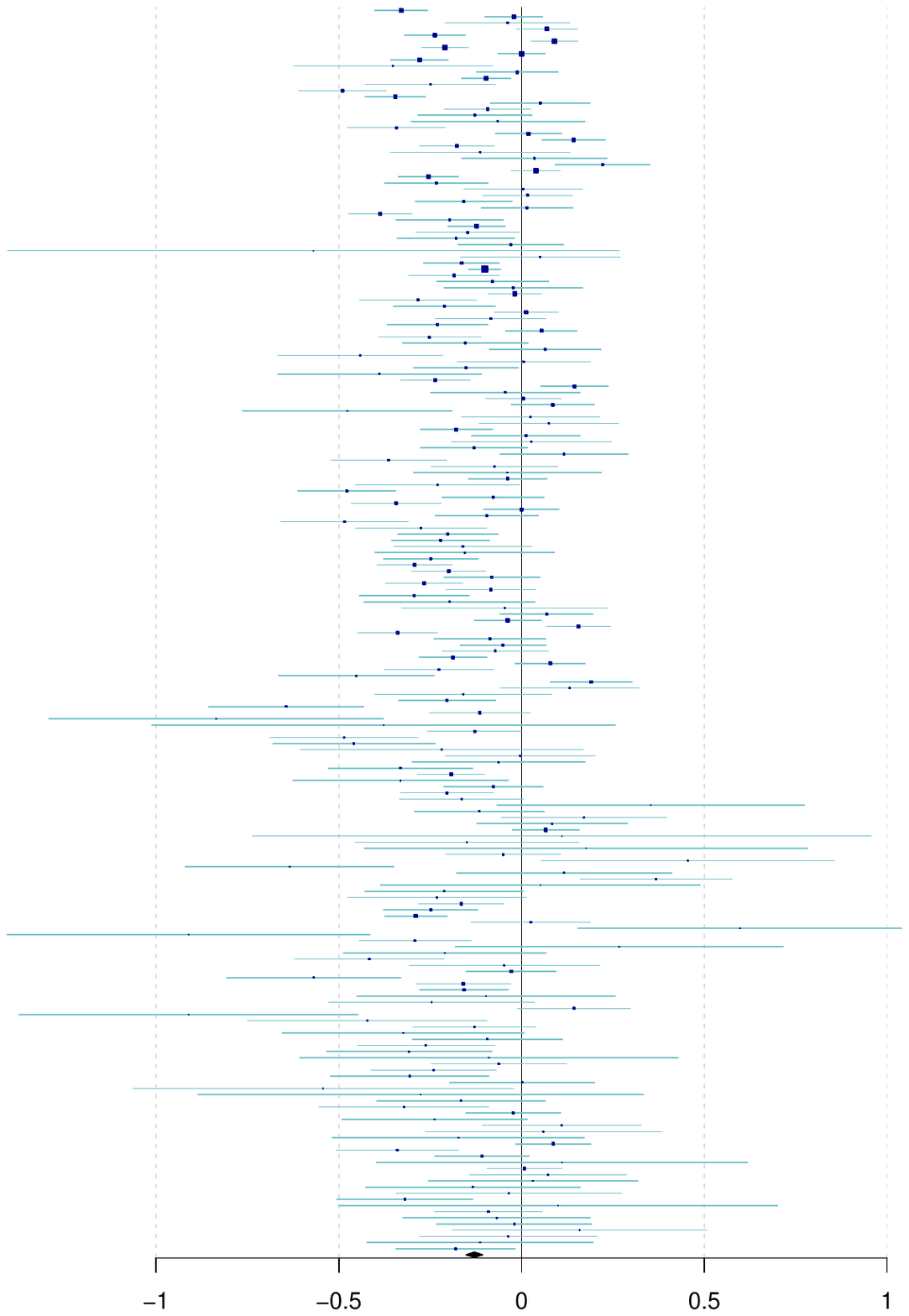}
    \caption{Estimates of the slope difference $\beta_3$}
            \vspace{1ex}
    \parbox{0.95\textwidth}{%
      \small
      \textit{Note.} This figure presents the forestplot of individual within-forecaster estimates of the coefficient $\beta_3$, i.e., the change in slope for time-varying event predictions. The forecasters are sorted (top to bottom) from most predictions made (total) to least predictions made. The horizontal axis refers to the coefficient's size. Whiskers indicate the 95\%-confidence intervals.
    }\label{FigForest4}
\end{figure}

\clearpage

\singlespacing

\bibliographystyle{apalike}
\bibliography{thebib}